\documentclass{tran-l}

\usepackage{tabularx}                
\usepackage{amsmath,amssymb, mathrsfs}
\usepackage{graphicx}
\usepackage[table]{xcolor}
\usepackage[utf8]{inputenc}  
\usepackage{booktabs} 
\usepackage{color,xcolor}

\copyrightinfo{200X}{American Mathematical Society}

\begin{document}

{\bf An Anchor-and-Propagate Scheme for Certified $O(n^2)$ Computation of Shapley Vectors in Cooperative Product Games}

\bigskip

\begin{center}
    {\large Jesse D Wei\textsuperscript{1}, Guo Wei\textsuperscript{2}}
\end{center}

\addtocounter{footnote}{2} 
\footnotetext[2]{\: Corresponding author. tel.: (910) 521 6582.\\
 {\it \hspace*{0.24in} E-mail addresses:} jesse@cs.unc.edu (J. Wei), guo.wei@uncp.edu (G. Wei).\\
 \hspace*{0.24in}  The original version is posted online at https://doi.org/10.48550/arXiv.2510.07572}

\bigskip
\hspace*{-0.25cm} $^1$ {Department of Computer Science, University of North Carolina at Chapel Hill,} \\
\centerline{Chapel Hill, NC, 27599, USA}
\hspace*{0.16cm} $^2$ {Department of Mathematics \& Computer Science, University of North Carolina at Pembroke,}\\
\centerline{Pembroke, North Carolina 28372, USA}

\newtheorem{theorem}{Theorem}[section]
\newtheorem{lemma}[theorem]{Lemma}
\newtheorem{corollary}[theorem]{Corollary}
\theoremstyle{definition}
\newtheorem{definition}[theorem]{Definition}
\newtheorem{example}[theorem]{Example}
\newtheorem{xca}[theorem]{Exercise}
\theoremstyle{remark}
\newtheorem{remark}[theorem]{Remark}
\numberwithin{equation}{section}

\bigskip

\noindent
{\bf Abstract}

For a general cooperative game, the Shapley values of the players need not be related,
but for cooperative product games they are closely linked. Identifying and exploiting
these relations is one purpose of this study. We consider a network of \(n\) devices
that participate independently with Bernoulli probabilities \((p_1,\dots,p_n)\), which
gives rise to the plausibility game \(T(S)=1-\prod_{i\in S}(1-p_i)\). The Shapley value
of device \(i\) has the exact integral representation
\[
\phi_i(T)=p_i\int_0^1\prod_{j\ne i}(1-p_jt)\,dt ,
\]
and, since \(T=1-w\) for the product game \(w(S)=\prod_{i\in S}(1-p_i)\), all results
transfer to product games with \(u_j\in[0,1]\) in \(w(S)=\prod_{j\in S}u_j\).

Our main contribution is an \emph{anchor-and-propagate} scheme: once \(\phi_1\) is
computed exactly, the remaining \(n-1\) values are obtained in \(O(n)\) total time,
with certified error bounds, from the exact identities
\[
\phi_i=c_i\,\phi_1+p_i(p_i-p_1)\,B_i
\qquad\text{and}\qquad
\phi_i=\phi_1+(p_i-p_1)\,J_i ,
\]
where \(c_i=p_i/p_1\), \(N_i=\{2,\dots,n\}\setminus\{i\}\),
\(B_i=\int_0^1t\prod_{j\in N_i}(1-p_jt)\,dt\) and
\(J_i=\int_0^1\prod_{j\in N_i}(1-p_jt)\,dt\); these satisfy \(B_i\le\tfrac12J_i\).
Let \(S_i=\sum_{j\in N_i}p_j\). For small \(S_i\), a Bonferroni expansion of \(B_i\)
truncated at order \(k\) has error at most \(S_i^{k+1}/((k+1)!\,(k+3))\). For large
\(S_i\), we introduce the estimator \(\hat B_i=G_1(S_i)-\tfrac12Q_iG_3(S_i)\), where
\(G_m(S)=\int_0^1t^me^{-St}\,dt\) and \(Q_i=\sum_{j\in N_i}p_j^2\); it satisfies a
two-sided error bound with relative error of order \(m_i^2/S_i^2\), where
\(m_i=\max_{j\in N_i}p_j\).

The anchor \(\phi_1\) is computed by Gauss--Legendre quadrature with \(\lceil n/2\rceil\)
nodes in \(O(n^2)\) operations or, avoiding irrational nodes, exactly in rational
arithmetic by polynomial expansion. The whole Shapley vector then costs
\(\tfrac32n^2+O(n)\) operations, about half the cost of the exact baseline and less than
that of QuadraSHAP. For \(n=100\), the scheme uses \(1.8\times10^4\) operations and
attains errors below \(3.4\times10^{-7}\), within a certified bound of
\(2.1\times10^{-6}\); permutation sampling (Castro et al.) with \(10^5\) permutations
needs \(4\times10^7\) operations for an error of \(7.9\times10^{-4}\), controlled only in
probability. We also give an anchor-free closed-form approximation
\(\phi_i\approx p_i\big[G_0(\sigma_i)-\tfrac12Q_{(i)}G_2(\sigma_i)\big]\), with
\(\sigma_i=\sum_{j\ne i}p_j\) and \(Q_{(i)}=\sum_{j\ne i}p_j^2\), which has an explicit
error bound. Finally, we give a self-contained derivation of Shapley's formula from its unanimity-game
decomposition, using combinatorial identities, the binomial theorem and the Beta function.

\bigskip

\bigskip




\bigskip
\noindent {\it ACM:} G.3, I.2.4.

\section{Known approximation methods to the Shapley value and the existing research gap}
\label{SecIntroduction}

The Shapley value is a solution concept in cooperative game theory that assigns a fair distribution of payoffs to players based on their marginal contributions. Since its introduction by Shapley (1953b), the value has become a cornerstone of cooperative game theory and has found numerous applications in economics, operations research, machine learning, and network analysis.

\subsection{Weighted and nonsymmetric Shapley values}

The weighted Shapley value originates in Shapley's Ph.D.\ dissertation (1953a). Owen (1968, 1972) studied it probabilistically; Weber (1988), Shapley (1981), Kalai and Samet (1987), and Hart and Mas-Colell (1989) provided axiomatizations of nonsymmetric values. Collectively, these works show that the Shapley value extends beyond the symmetric case while retaining efficiency, linearity, and the dummy-player property.

\subsection{Derivation methods and historical context}

Over the years, researchers have developed multiple methods to derive or compute the Shapley value, each offering unique insights. The main calculation and derivation methods in the literature can be summarized as follows. For a deep dive into Shapley's original method, the reader is referred to the original references (Shapley 1953b; Maschler, Solan, and Zamir 2013; Osborne and Rubinstein 1994; Rota 1964, 1988; Winter 2002).

\medskip

\paragraph{Original axiomatic derivation (Shapley 1953b).}
Shapley derived his value using four axioms.

\medskip
Efficiency: the total payoff is fully distributed.

Symmetry: players with identical marginal contributions receive equal payoffs.

Dummy player (carrier): a player who contributes nothing to any coalition gets nothing.

Additivity: the Shapley value of a sum of games is the sum of the Shapley values of the individual games.

\medskip
The unique solution satisfying these axioms is
\[
\phi_i(v) = \sum_{S \subseteq E \setminus \{i\}} \frac{s! (n - s - 1)!}{n!} \left(v(S \cup \{i\}) - v(S)\right),
\]
where \(E\) is the set of all \(n\) players, \(v\) is the characteristic function, and \(s\) is the size of coalition \(S\).

\medskip

\paragraph{Permutation-based approach.}
The Shapley value can be interpreted as the average marginal contribution of a player across all possible orderings of players:
\[
\phi_i(v) = \frac{1}{n!} \sum_{\pi \in \Pi(E)} \left[v(P_\pi^i \cup \{i\}) - v(P_\pi^i)\right],
\]
where \(\Pi(E)\) is the set of all permutations of \(E\), and \(P_\pi^i\) is the set of players preceding \(i\) in permutation \(\pi\).

\medskip

\paragraph{Potential function method (Hart and Mas-Colell 1989).}
This method introduces a potential function \(P\) such that
\[
\phi_i(v) = P(E,v) - P(E \setminus \{i\}, v).
\]
The potential function is recursively defined, ensuring consistency with the Shapley value.

\medskip

\paragraph{Harsanyi dividends (Harsanyi 1963).}
The Shapley value can be expressed in terms of dividends for each coalition:
\[
\phi_i(v) = \sum_{\substack{S \subseteq E \\ i \in S}} \frac{\Delta_S(v)}{|S|},
\]
where \(\Delta_S(v)\) is the Harsanyi dividend of coalition \(S\), defined recursively.

\medskip

\paragraph{Multilinear extension method (Owen 1972).}
This approach replaces the discrete game by a continuous extension. For a game
\(v\) on \(N=\{1,\dots,n\}\), the \emph{multilinear extension} of \(v\) is the
function \(f:[0,1]^n\to\mathbb{R}\),
\[
f(x)=\sum_{S\subseteq N} v(S)\prod_{i\in S}x_i\prod_{j\in N\setminus S}(1-x_j),
\]
which is affine in each coordinate and agrees with \(v\) at the vertices of the
unit cube, \(f(\mathbf 1_S)=v(S)\). Owen showed that, for every cooperative game,
the Shapley value is obtained by integrating the partial derivative of \(f\)
along the main diagonal of the cube:
\[
\phi_i(v) = \int_0^1 \frac{\partial f}{\partial x_i}(t,t,\dots,t)\,dt .
\]
The proof rests on the Beta-integral identity
\[
\frac{s!\,(n-s-1)!}{n!}=\int_0^1 t^{s}(1-t)^{n-s-1}\,dt ,
\]
which turns the combinatorial weights in Shapley's original formula into an
integral over \(t\in[0,1]\). For the plausibility game
\(T(S)=1-\prod_{i\in S}(1-p_i)\), the multilinear extension is
\(f(x)=1-\prod_{j\in N}(1-p_jx_j)\), so
\(\partial f/\partial x_i(t,\dots,t)=p_i\prod_{j\ne i}(1-p_jt)\), and Owen's
formula reduces to the integral representation
\(\phi_i(T)=p_i\int_0^1\prod_{j\ne i}(1-p_jt)\,dt\) used throughout this paper.

\medskip

\paragraph{Probabilistic interpretation (Weber 1988).}
The Shapley value can be seen as the expected marginal contribution when players join a coalition in a random order, where each ordering is equally likely.

\medskip

\paragraph{Linear programming formulation (Charnes, Golany, Keane, and Rousseau 1988).}
The Shapley value can be derived as the solution to a linear program minimizing a certain imbalance criterion while satisfying efficiency and symmetry.

\medskip

\paragraph{Random order value equivalence (Weber 1988).}
The Shapley value is the unique random-order value that satisfies symmetry and efficiency.

\medskip

\paragraph{Aumann--Shapley pricing (Aumann and Shapley 1974).}
For non-atomic games (continuum of players), the Shapley value is derived via differentiation of a path integral.

Beyond the Shapley value itself, higher-order interaction effects have also been studied. Murofushi and Soneda (1993) introduced the interaction index for a pair of players. Grabisch (1997) generalized the index to any coalition \(S\). These interaction indices provide a finer description of how coalitions of players jointly affect the game, and they are naturally connected to the M\"obius transform of the characteristic function. Interaction indices of coalitions are not covered in this work, but they motivate the structural study of product games and their M\"obius transforms in later sections.

\medskip

\subsection {Approximation and sampling methods.}
For general games, exact computation of the Shapley value requires exponentially many coalition
values, so approximation algorithms are essential. A landmark in this literature is the work of
Castro, G\'omez, and Tejada (2009), who use Shapley's (1953b) representation of the value as the
expected marginal contribution over uniformly random orderings of the players to obtain a
Monte Carlo estimator that runs in polynomial time and comes with a probabilistic error bound.
Castro et al.\ (2017) improved its accuracy by stratified sampling with optimal allocation of the
samples. Related sampling methods have been developed for specific applications, notably the
valuation of training data in machine learning (Ghorbani and Zou 2019; Jia et al.\ 2019b).
These general-purpose solvers should be distinguished from feature-attribution methods for
machine-learning models, such as SHAP and its variants KernelSHAP and DeepSHAP (Lundberg and
Lee 2017). These methods define a game from a trained model and a background dataset, the value
of a set of features being the model's expected output when only those features are known, and
then approximate the Shapley values of that game; DeepSHAP is further specialized to deep neural
networks.

\subsection{Existing gaps and scope of this study}
\label{SecScope}

\paragraph{1. The gap.}
Computing the exact Shapley vector of a general cooperative game requires exponential time.
Sampling methods run in polynomial time, but their error is controlled only in probability and
decreases only like the inverse square root of the sample size. For a general game this is
unavoidable, since the Shapley values of different players need not be related by any simple
pattern. For structured classes, such as cooperative product games, the values are closely related,
and these relations have not been exploited.

The network game studied here, \(T(S)=1-\prod_{j\in S}(1-p_j)\), is defined analytically, with no
predictive model or dataset, so the model-based constructions of SHAP do not apply. It is
equivalent to a cooperative product game (item~3 and Section~\ref{SecCorrespondence}), and for product games each \(\phi_i\) is a one-dimensional
integral of a polynomial of degree at most \(n-1\). Gauss--Legendre quadrature with
\(\lceil n/2\rceil\) nodes evaluates it exactly, and sharing the node products gives the whole vector
in \(O(n^2)\) operations. The exact vector is therefore accessible in \(O(n^2)\) operations, once the nodes and weights are given. What remains
open is threefold: these methods compute each value independently of the others, ignoring the
relations among them; their exactness and their stated costs rest on irrational nodes and weights
that are taken as given; and their approximate regime has no error bound that can be evaluated for a
given game.

\medskip

\paragraph{2. Castro et al.'s study.}
Castro et al.\ (2009, 2017) provide the leading general-purpose sampling framework. Their stratified
method treats \(v\) as an oracle, samples random permutations, and reduces the variance by
stratifying on the player's position with Neyman allocation of the samples. None of their five
benchmarks (two voting games, an airport game, a shoes game and a minimum spanning tree game) is a
product game, and their method does not exploit algebraic structure. Sampling could be applied to
\(T\), but it is not needed; Section~\ref{SecOverallComparison} compares the two approaches.

\medskip

\paragraph{3. QuadraSHAP}
Mohammadi et al.\ (2026) compute Shapley values of product games \(v(S)=\prod_{j\in S}u_j\) by
Gauss--Legendre quadrature. In their paper \(d\) denotes the number of players (features) and \(n\)
the number of training samples; training samples play no role here, so we write \(n\) for the number
of players and restate their results in this notation. The Shapley value is
\[
\phi_i(v)=(u_i-1)\int_0^1 g_i(t)\,dt,\qquad g_i(t)=\prod_{j\ne i}\bigl((1-t)+tu_j\bigr),
\]
and with \(u_j=1-p_j\) one has \(v=1-T\) and \(\phi_i(T)=-\phi_i(v)\). With \(m_q\) nodes, their
results are: exactness for \(m_q\ge\lceil n/2\rceil\) (their Proposition~3(i)); the whole vector in
\(O(n\,m_q)\) operations by forming each node product once and dividing out one factor in log space,
and in parallel time \(O(\log n)\) (their Proposition~4); and, their headline result, geometric decay
of the error in \(m_q\) for \(m_q<\lceil n/2\rceil\), with a few hundred nodes sufficing in their
experiments with up to \(10^4\) features (their Proposition~3(ii)).

These results are statements about exact arithmetic with the nodes and weights in hand, which limits
them in three ways (Section~\ref{sec:quadrashap}):
\begin{enumerate}
\item \emph{Exactness.} For \(m_q\ge2\) the nodes are irrational, except the middle node
\(\tfrac12\) for odd \(m_q\), and for \(m_q\ge4\) so are the weights
(Remark~\ref{rem:nodes-rational}); the exact regime is exact only up to their rounding.
\item \emph{Cost.} The counts exclude the nodes and weights, which cost \(O(m_q^2)\) by
Golub--Welsch (Golub and Welsch 1969) and \(O(m_q^3)\) with the dense solvers of common libraries,
that is \(O(n^3)\) in the exact regime, unless an \(O(m_q)\) method is used (Hale and Townsend
2013; Bogaert 2014) (Remark~\ref{rem:quad-assumptions}).
\item \emph{Certified accuracy.} The approximate-regime bound involves an unspecified maximum
\(M_\rho\) on a Bernstein ellipse and a constant \(C_\rho\) that must be corrected before the bound can
be evaluated (Section~\ref{sec:quad-approx}).
\end{enumerate}
Of the node-free alternatives, polynomial expansion is exact in rational arithmetic but, applied
player by player, costs \(O(n^2)\) per value; the FFT--multipoint method obtains the whole vector in
\(O(n\log^2n)\) operations in exact arithmetic but is unstable in floating point
(Section~\ref{SecFFTMultipoint}).

\medskip

\paragraph{4. Scope of our study.}
We study networks whose components join or fail independently according to inhomogeneous Bernoulli
trials, a standard baseline in reliability and risk analysis, with the analytically known game
\(T\) and the representation \(\phi_i=p_i\int_0^1\prod_{j\ne i}(1-p_jt)\,dt\). Our approach
exploits the relations among the Shapley values: the identities of Section~\ref{SecAnchor},
\[
\phi_i=c_i\phi_1+p_i(p_i-p_1)B_i
\qquad\text{and}\qquad
\phi_i=\phi_1+(p_i-p_1)J_i,
\]
where \(c_i=p_i/p_1\), \(N_i=\{2,\dots,n\}\setminus\{i\}\),
\[
B_i=\int_0^1 t\prod_{j\in N_i}(1-p_jt)\,dt
\qquad\text{and}\qquad
J_i=\int_0^1\prod_{j\in N_i}(1-p_jt)\,dt ,
\]
express every \(\phi_i\) through a single anchor value \(\phi_1\) and a one-dimensional integral.
Once \(\phi_1\) is known, the remaining \(n-1\) values follow in \(O(n)\) operations from
closed-form estimates of \(B_i\), with certified error bounds (Section~\ref{sec:error}).
As a result:
\begin{itemize}
\item the nodes and weights of Gauss--Legendre quadrature, if used at all, enter only through
\(\phi_1\), which can also be computed by polynomial expansion, which needs no nodes;
\item the whole vector costs about \(\tfrac32n^2\) operations, half the \(3n^2\) of the exact
prefix/suffix baseline (Section~\ref{sec:opcounts});
\item each remaining value has an error bound that can be evaluated directly;
\item with \(J_i\) computed by polynomial division, the second identity gives the exact whole vector
without nodes, in rational arithmetic, in about \(5n^2\) operations.
\end{itemize}
The price of the fast scheme is that \(\phi_i\), \(i\ge2\), are certified estimates rather than exact
values. For the example with \(n=100\) (Section~\ref{SecAccuracy100}), it is more accurate than
permutation sampling with \(10^5\) permutations by a factor of more than \(2{,}000\), at less than
\(1/2{,}000\) of the arithmetic cost. Whether the exact whole vector can be computed stably in
floating point in fewer than \(O(n^2)\) operations remains open: the only known route below
\(O(n^2)\), the FFT--multipoint method, achieves this only in exact arithmetic.

\medskip
\subsection*{Organization of the Paper}

Section~\ref{sec:related} reviews what is known about relations among the Shapley values of a game.
Sections~\ref{Sec2}--\ref{SecConnection} establish the theoretical foundation: the random set
representation of characteristic functions (Section~\ref{Sec2}), the correspondence between
inhomogeneous Bernoulli games, belief functions and cooperative product games
(Section~\ref{Sec4}), the M\"obius transform of product games (Section~\ref{SecMobius}), and the
connection with the product-game literature (Section~\ref{SecConnection}).

Section~\ref{SecAnchor} contains the main results. It derives the anchor identity, which expresses
each \(\phi_i\), \(i\ge2\), through the anchor value \(\phi_1\) and a one-dimensional integral
\(B_i\); establishes the relations and bounds among the integrals \(B_i\), \(I_i\) and \(J_i\); and
derives a second anchor identity, \(\phi_i=\phi_1+(p_i-p_1)J_i\) (Section~\ref{SecSecondAnchor}),
which also gives the exact whole vector without quadrature nodes. Section~\ref{sec:error} gives the
error analysis: estimators of \(B_i\) for small and large \(S_i\), certified error bounds, a rule
for choosing between them, and an \(O(n)\) cost for all remaining values. It also treats \(\phi_i\)
as a function of \(p_1\) and gives a closed-form approximation that needs no anchor
(Section~\ref{sec:phi-p1}).

Section~\ref{sec:quadrashap} compares these methods with QuadraSHAP (Mohammadi et al.\ 2026): the
assumptions behind its operation counts, the rationality of the Gauss--Legendre nodes and weights,
exact operation counts for all methods, an explicit and computable form of the error bound of its
Proposition~3(ii) for our game, and how the anchor scheme combines with its approximate regime and
how the anchor should be chosen (Section~\ref{sec:anchor-choice}).
Section~\ref{SecWorkedExamples} illustrates both anchor identities on a small instance and a
network-security application, and compares the methods of this paper with the permutation sampling
of Castro et al.\ (2009, 2017) and the FFT--multipoint method on an example with \(n=100\).
Appendix~I (Section~\ref{Sec1}) gives a self-contained combinatorial and integral derivation of the
Shapley value formula.

\section{Related work: relations among the Shapley values of a game}\label{sec:related}

We review what is known about obtaining the Shapley values of the remaining $n-1$ players from one
of them, first for general cooperative games and then for special classes. Throughout,
$\phi_i(v)$ denotes the Shapley value of player $i$ in a game $v$ on
$N=\{1,\dots,n\}$, and $T(S)=1-\prod_{j\in S}(1-p_j)$ is the hitting game studied in this paper.

\subsection{General cooperative games}

\paragraph{One value alone carries no information about the others.}
The only constraint linking all Shapley values of a general game is efficiency,
$\sum_{i\in N}\phi_i(v)=v(N)$. Hence $n-1$ values determine the last one, but a single value
determines nothing more: adding an additive game $a(S)=\sum_{j\in S}x_j$ with $x_1=0$ and
$\sum_jx_j=0$ leaves $\phi_1$ and $v(N)$ unchanged while shifting $\phi_j$ by an arbitrary $x_j$,
$j\ge2$. Any result of the form ``one value determines (or approximates) the others'' therefore
requires structure in the game.

\paragraph{Exact relations between values.} Several classical identities relate the values of
different players.
\begin{itemize}
\item \emph{Pairwise differences.} A direct consequence of the definition is
\begin{equation}\label{eq:rw-diff}
\phi_i(v)-\phi_j(v)=\sum_{S\subseteq N\setminus\{i,j\}}\frac{s!\,(n-2-s)!}{(n-1)!}\,
\big[v(S\cup\{i\})-v(S\cup\{j\})\big],\qquad s=|S| .
\end{equation}
For the hitting game (see Section~\ref{Sec4}),
\[
T(S\cup\{i\})-T(S\cup\{j\})=(p_i-p_j)\prod_{r\in S}(1-p_r),
\]
and the Beta-integral representation of the weights gives the exact anchor formula
\begin{equation}\label{eq:rw-anchor}
\phi_i(T)=\phi_1(T)+(p_i-p_1)\int_0^1\prod_{r\ne1,i}(1-p_rt)\,dt.
\end{equation}
This is an alternative to the decomposition $\phi_i=c_i\phi_1+p_i(p_i-p_1)B_i$ used in this paper (see Section~\ref{SecAnchor}).

\item \emph{Balanced contributions} (Myerson 1980; see also B\'eal et al.\ 2015):
\[
\phi_i(N,v)-\phi_i(N\setminus\{j\},v)=\phi_j(N,v)-\phi_j(N\setminus\{i\},v),
\]
which relates the values of two players through their values in subgames.
\item \emph{Potential} (Hart and Mas-Colell 1989): there is a unique function $P$
on games such that $\phi_i(N,v)=P(N,v)-P(N\setminus\{i\},v)$ for all $i$. All $n$ values are then
obtained from the potentials of $N$ and of its $n$ subgames $N\setminus\{i\}$.
\item \emph{Multilinear extension} (Owen 1972): if $f$ is the multilinear extension of
$v$, then
\[
\phi_i(v)=\int_0^1\partial_if(t,\dots,t)\,dt .
\]
For the hitting game this gives $\phi_i(T)=p_i\int_0^1\prod_{j\ne i}(1-p_jt)\,dt$, the starting point
of our analysis.
\end{itemize}
These identities do not by themselves reduce the cost of computing the values of a general game,
since \eqref{eq:rw-diff}, the subgame values and the potential all involve exponentially many
coalitions.

\subsection{Special classes of games}

\paragraph{Airport games.} For the airport game of Littlechild and Owen (1973),
with players ordered by their cost requirements $c_1\le\dots\le c_n$ and $c_0=0$,
\[
\phi_i=\sum_{k=1}^{i}\frac{c_k-c_{k-1}}{n-k+1}
\]
(see also M\'arkus, Pint\'er, and Radv\'anyi 2011). Equivalently, $\phi_i=\phi_{i-1}+(c_i-c_{i-1})/(n-i+1)$: each value is
obtained exactly from the previous one.

\paragraph{Nearest-neighbour data valuation.} For the utility of an unweighted $K$-nearest-neighbour
classifier, Jia et al.\ (2019a) show that, with training points $\alpha_1,\dots,\alpha_N$
sorted by distance to the test point,
\[
s_{\alpha_N}=\frac{\mathbf 1(y_{\alpha_N}=y_{\mathrm{test}})}{N},\qquad
s_{\alpha_i}=s_{\alpha_{i+1}}+\frac{\mathbf 1(y_{\alpha_i}=y_{\mathrm{test}})-\mathbf 1(y_{\alpha_{i+1}}=y_{\mathrm{test}})}{K}\cdot\frac{\min\{K,i\}}{i},
\]
where $\mathbf 1(\cdot)$ denotes the indicator function, so that
$\mathbf 1(y_{\alpha_i}=y_{\mathrm{test}})$ equals $1$ if the $i$-th nearest
training point $\alpha_i$ has the same label as the test point and $0$ otherwise.

This is the closest published analogue of the approach of this paper: one value serves as an anchor
and all others follow by an exact recursion, at total cost $O(N\log N)$. The recursion is exact and
specific to nearest-neighbour utilities.

\bigskip

\paragraph{Product games.} Dehez (2024) studies the product game
$v_a(S)=\prod_{i\in S}a_i-1$ with $a_i\ge1$. He shows that
$\phi_i(v_a)=(a_i-1)\,f(a_{-i})$ for a symmetric function $f$, that
$\phi_i(v_a)-\phi_j(v_a)=(a_i-a_j)\,[\,\cdots]$ with a positive bracket, so that $\phi_i=\phi_j$ if and
only if $a_i=a_j$, and he characterizes the Shapley value on this class by proportionality, efficiency
and symmetry. Closed forms are given for small $n$; for example $a=(2,4,6)$ gives
$\phi(v_a)=(10,17,20)$.

The corrected product game $v_P$ of this paper is Dehez's game with $a_j=1/(1-p_j)$. Since the
Shapley value is self-dual and $v_P=\frac{c}{Q}V$, where $V$ is the dual of $T/c$,
$\phi_i(v_P)=\phi_i(T)/Q$ with $Q=\prod_{j}(1-p_j)$. Our integral representation then gives
Dehez's function in compact form,
\[
f(a_{-i})=\int_0^1\prod_{j\ne i}\big(a_j-(a_j-1)\,t\big)\,dt ,
\]
which reproduces his example $(10,17,20)$. Consequently the exact $O(n^2)$ algorithm and the $O(n)$
estimators with certified error bounds developed here apply to Dehez's product games.

\medskip

\paragraph{Reliability games.} For independent and identically distributed (more generally,
exchangeable) component lifetimes, the Barlow--Proschan importance index of a component equals the
Shapley--Shubik value of the system's structure function; Marichal and Mathonet (2013)
extend the index to dependent lifetimes and show that it is a probabilistic value.

\medskip

\paragraph{Weighted voting games.} Shapley--Shubik indices of weighted voting games can be computed by
generating functions, and they respect the order of the weights. We are not aware of results that
derive the remaining indices from a single one.

\subsection{Approximation methods that exploit relations among values}

Jia et al.\ (2019b) propose a group-testing estimator for general utilities: sampled
coalitions are used to estimate the pairwise differences $s_i-s_j$, and all values are then recovered
from these differences together with the single efficiency constraint, with a provable error bound.
This is the general-game counterpart of the approach taken here, where the exactly computed
$\phi_1$ serves as the anchor and the differences have closed forms. Generic sampling estimators,
such as that of Castro, G\'omez, and Tejada (2009), estimate each value separately and do not
use relations between players.

\subsection{Position of this work}

In summary:
\begin{enumerate}
\item For general games, a single Shapley value does not determine the others; known relations
(pairwise differences, balanced contributions, the potential) still require exponentially many
coalition values.
\item For special classes, exact recursions are known for airport games and for nearest-neighbour
utilities, and structural results (proportionality, ordering) for product games.
\item For the hitting game $T$, and equivalently for Dehez's product game $v_P$ and the belief function
$V$, this paper provides an exact anchor representation, closed-form estimators of the remaining
$n-1$ values from one exactly computed value, certified error bounds, and $O(n)$ cost after the
anchor. We have not found these results in the literature.
\end{enumerate}

\section{Random set representation of characteristic functions} \label{Sec2}
Let $E = \{1, \dots, n\}$ be the set of $n$ players (i.e. devices). The characteristic function of a game of these $n$ players is a real-valued set function defined on the power set $\mathcal {F}$ of $E$, i.e. $v: \mathcal {F} \to [0,1]$ satisfying

\smallskip

i) $v(\emptyset) = 0, v(E) = 1$. (Normalization)

ii) If $A, B \in \mathcal {F}, A \subseteq B$, then $v(A) \leq v(B)$.  (Monotonicity)

\bigskip

\begin{definition}[Shapley Value]
The Shapley value of a player $i \in E$ in the game  $v$    is defined as
\[
\varphi_i(v)
=
\sum_{S \subseteq E \setminus \{i\}}
\frac{|S|!\,(n-|S|-1)!}{n!}
\bigl[v(S \cup \{i\}) - v(S)\bigr].
\]

Equivalently,
\[
\varphi_i(v)
=
\frac{1}{n!}
\sum_{\pi \in \Pi(E)}
\Bigl(
v\bigl(P_i^\pi \cup \{i\}\bigr) - v\bigl(P_i^\pi\bigr)
\Bigr),
\]
where $\Pi(E)$ denotes the set of all permutations of $E$, and $P_i^\pi$ is the set of players that precede $i$ in the ordering $\pi$.
\end{definition}

\begin{definition}[M\"obius Coefficients]
The M\"obius transform of $v$ is given by
\[
m_v(S) = \sum_{T \subseteq S} (-1)^{|S| - |T|} v(T).
\]
\end{definition}

\begin{definition}[Random Set; Choquet 1954]
A random set $X$ is a random element taking values in the space $\mathcal {F}$ of closed subsets of $E$, equipped with the discrete topology when $E$ is finite. Its Choquet capacity (covering function) is
\[
v(S) = \mathbb{P}(X \cap S \neq \emptyset), \; S \subseteq E.
\]
\end{definition}

\begin{theorem}[Belief-function representation; Shafer 1976]
If $v$ satisfies $v(E) = 1$ and $m_v(S) \geq 0$ for all $S \subseteq E$, then there exists a random set $X$ on $E$, unique in distribution, such that
\[
v(S) = \mathbb{P}(X \subseteq S) \quad \forall S \subseteq E.
\]
Here $v$ is a \emph{belief function} (containment function) in the sense of Dempster-Shafer theory.
\end{theorem}

\begin{remark}[Belief vs.\ Plausibility]
The belief function $v(S) = \mathbb{P}(X \subseteq S)$ and the plausibility (hitting) function
\[
u(S) = \mathbb{P}(X \cap S \neq \emptyset) = 1 - v(S^c)
\]
are conjugates. Both are determined by the same mass function $m_v(T)=\mathbb P(X=T)$ of the random set (although the M\"obius transform of $u$ differs from $m_v$ in sign pattern), and the Shapley values of $v$ and $u$ coincide whenever $v(E) = u(E) = 1$. {\it In this paper we begin with the plausibility (hitting) function} $T(S) = \mathbb{P}(X \cap S \neq \emptyset)$, which is the natural characteristic function for a random network of devices. Whenever we invoke the theorem above, we mean the belief function $v$, not $T$.
\end{remark}

The three-step proof given below establishes the belief-function representation. The corresponding plausibility statement follows from conjugacy.

\begin{enumerate}
    \item M\"obius positivity implies belief measure:
    A function $v$ with non-negative M\"obius coefficients is a belief function in Dempster-Shafer theory.

    \item Constructing the random set:
    Define the probability distribution:
    \[
    \mathbb{P}(X = T) = m_v(T) \quad \forall T \subseteq E.
    \]
    This is valid since $\sum_{T \subseteq E} m_v(T) = v(E) = 1$.

    \item Verification:
    The containment probability satisfies:
    \[
     \mathbb{P}(X \subseteq S) = \sum_{T \subseteq S} m_v(T) = v(S).
    \]
    Taking complements, this is equivalent to the plausibility representation
    \[
    u(S) = \mathbb{P}(X \cap S \neq \emptyset) = 1 - v(S^c).
    \]
\end{enumerate}
The second equality in Step 3 is the definition of the distribution in Step 2, and the plausibility form follows from $\{X\cap S\neq\emptyset\}=\{X\subseteq S^c\}^c$.
\bigskip

 Therefore, any totally positive characteristic function with $v(E) = 1$ can be represented as the belief function (containment) of a random set; its conjugate $u(S)=1-v(S^c)$ is then the plausibility (hitting) probability of the same random set.

\section{From    inhomogeneous Bernoulli games to cooperative product games} \label{Sec4}
\subsection {Inhomogeneous Bernoulli games}
Let $E = \{1,2, ..., n\}$, with discrete topology, and $\mathcal {F} = \{F| F \subseteq E\}$. For $S \subseteq E$, let $\mathcal {F}_S = \{F \in \mathcal {F} | F \cap S \neq \emptyset\}$, the set of all subsets of $E$ that intersect with $S$. Suppose that $X$ is a random set defined on $E$, and $T$ is the Choquet capacity (covering function) of $X$, i.e. $T(S) = \mathbb{P}(X \cap S \neq \emptyset)$, the probability that $X$ intersects with $S$. Assume that ${\bf P}$ is the associated probability measure defined on the hyperspace $\mathcal {F}$ with the Fell topology (actually discrete as $E$ is finite), satisfying {\bf P}$(\mathcal {F}_S) = T(S)$.


Also, define
\[
\mathcal {F}^S
=
\{F \subseteq E \mid F \cap S = \emptyset\}.
\]

The events $(X \in \mathcal {F}^S)$ and $(X \in \mathcal {F}^{\{i\}})$, where $i \notin S$, are generally not independent. Indeed, independence would require
\[
{\bf P}\!\left(
X\in\mathcal F^S\cap\mathcal F^{\{i\}}
\right)
=
{\bf P}(X\in\mathcal F^S)\,
{\bf P}(X\in\mathcal F^{\{i\}}).
\]
Since
\[
\mathcal F^S \cap \mathcal F^{\{i\}}
=
\mathcal F^{S \cup \{i\}},
\]
this condition is equivalent to
\[
1-T(S\cup\{i\})
=
(1-T(S))(1-T(\{i\})).
\]

Thus, independence holds precisely when the random set $X$
treats the exclusion of $S$ and the exclusion of $i$ as independent events.
This occurs, for example, in random sets generated by independent element inclusions, such as inhomogeneous Bernoulli models in which each element $j \in E$ is included in $X$ independently with probability $p_j$.

\bigskip

When the elements of $X$ are independent, the Choquet capacity $T$ of $X$ takes a particularly multiplicative form. Based on this fact, we will  derive the concrete form of $T$ for the inhomogeneous Bernoulli trials below.

The key assumption in this study is the independence below:

    Each element $i \in E$ is included in $X$ independently with probability $p_i$:
    \[
    \mathbb{P}(i \in X) = p_i, \quad \mathbb{P}(i \notin X) = 1 - p_i.
    \]
The probability that $X$ does not intersect $S$ (i.e., $X \cap S = \emptyset$) is the product of the probabilities that each $j \in S$ is not in $X$:
    \[
    \mathbb{P}(X \cap S = \emptyset) = \prod_{j \in S} (1 - p_j).
    \]

\noindent
{Choquet Capacity $T(S)$.}
By definition, $T(S)$ is the probability that $X$ intersects $S$:
\[
T(S) = \mathbb{P}(X \cap S \neq \emptyset) = 1 - \mathbb{P}(X \cap S = \emptyset).
\]
Substituting the independent case gives
\[
T(S) = 1 - \prod_{j \in S} (1 - p_j), \quad S \subseteq E.
\]

\subsection{The Correspondence Between Inhomogeneous Bernoulli Plausibility Games, Belief Functions, and Cooperative Product Games}
\label{SecCorrespondence}

This subsection makes precise the correspondence
\[
T
\;\longrightarrow\;
\widehat{T}
\;\longrightarrow\;
V
\;\longrightarrow\;
v_P
\]
that runs through the paper, and clarifies in which direction the correspondence is reversible. The main result is stated as a theorem, followed by a short proof, a counterexample for the reverse direction, and a summary.

\subsubsection{Setup and Notation}

Recall \(E=\{1,\dots,n\}\) and let \(p_1,\dots,p_n\in(0,1)\).
The plausibility (hitting) function of the inhomogeneous Bernoulli random set \(X\) is
\[
T(S) = 1 - \prod_{j\in S}(1-p_j), \; S\subseteq E,
\]
where each \(j\in E\) is included in \(X\) independently with probability \(p_j\). Define
\[
Q := \prod_{j\in E}(1-p_j),
\qquad
c := 1 - Q = T(E),
\]
\[
a_j := \frac{1}{1-p_j},
\qquad
v_P(S) := \prod_{j\in S} a_j - 1,
\]
\[
B(S) := \prod_{j\in S^c}(1-p_j),
\qquad
\widehat{T}(S) := \frac{T(S)}{T(E)},
\qquad
V(S) := 1 - \widehat{T}(S^c).
\]
Equivalently,
\[
B(S) = Q\prod_{j\in S} a_j,
\qquad
V(S) = \frac{B(S)-Q}{c} = \frac{Q}{c}\,v_P(S).
\]

The chain can be written in five equivalent stages (Table~\ref{tab:objects} and the note following the table).

\begin{table}[htbp]
\centering
\caption{The five game-theoretic objects and their roles in the network model.}
\label{tab:objects}
\begin{tabular}{@{}c l l p{5.6cm}@{}}
\toprule
Stage & Object & Definition & Role \\
\midrule
1 & \(T\) & \(1-\prod_{j\in S}(1-p_j)\) & Plausibility (hitting) function of \(X\) \\
2 & \(\widehat{T}\) & \(T/c\) & Normalized plausibility function \\
3 & \(B\) & \(\prod_{j\in S^c}(1-p_j)\) & Unnormalized belief function, \(B(S)=\mathbb{P}(X\subseteq S)\) \\
4 & \(V\) & \((B-Q)/c\) & Normalized belief function, \(V=1-\widehat{T}(\cdot^{\,c})\) \\
5 & \(v_P\) & \(B/Q-1\) & Corrected product game ($v_P(\emptyset) = \prod_{j\in \emptyset} a_j - 1 = 1-1 = 0$; Dehez 2024) \\
\bottomrule
\end{tabular}
\end{table}

\noindent
Note: \(B\) is not an independent object: it is the
unnormalized version of \(V\) and the shifted-and-scaled version of
\(v_P\). Specifically,
\[
B(S) = Q\bigl(v_P(S)+1\bigr),
\qquad
V(S) = \frac{B(S)-Q}{c} = \frac{Q}{c}\,v_P(S).
\]

\subsubsection{The Correspondence Theorem}

\begin{theorem}[The Correspondence Theorem]
\label{ThmCorrespondence}
Let \(T\), \(B\), \(V\), and \(v_P\) be as in the Setup, with
\[
Q:=\prod_{j\in E}(1-p_j),\qquad c:=1-Q,\qquad
\widehat T(S):=\frac{T(S)}{c},\qquad
V(S):=1-\widehat{T}(S^c)=\frac{B(S)-Q}{c}.
\]
The following hold.

\begin{enumerate}
\item[(a)] (Forward direction, true.) Every inhomogeneous Bernoulli
plausibility game \(T\) gives rise to

\begin{enumerate}
  \item[(i)] a normalized plausibility capacity \(\widehat{T}\) with
  \(\widehat{T}(\varnothing)=0\), \(\widehat{T}(E)=1\), which is
  monotone but not, in general, a belief function;

  \item[(ii)] the conjugate \(V(S)=1-\widehat{T}(S^c)\), which is a
  \emph{belief function}: \(V(\varnothing)=0\), \(V(E)=1\), and
  \(V\) is completely monotone (equivalently, its M\"obius transform is
  nonnegative);

  \item[(iii)] the unnormalized belief capacity
  \(B(S)=\prod_{j\in S^c}(1-p_j)\), which satisfies \(B(\varnothing)=Q\)
  and \(B(E)=1\), and hence is \emph{not} a belief function
  (a belief function requires value \(0\) at \(\varnothing\));
  and

  \item[(iv)] the corrected cooperative product game
  \(v_P(S)=\prod_{j\in S}a_j-1\) with \(a_j=(1-p_j)^{-1}\), which
  satisfies \(v_P(\varnothing)=0\).
\end{enumerate}

\noindent
These objects are related, for \emph{every} \(S\subseteq E\), by
\[
V(S)=1-\widehat{T}(S^c)=\frac{B(S)-Q}{c}=\frac{Q}{c}\,v_P(S),
\qquad
B(S)=Q\bigl(v_P(S)+1\bigr),
\qquad
v_P(S)=\frac{c}{Q}\,V(S).
\]
In particular, \(B\) is a positive multiple of \(v_P+1\) on all of
\(\mathcal {F}\), and \(V\) is a positive multiple of \(v_P\) on all of \(\mathcal{F}\).

\item[(b)] (Bijection on the product-structured class.)
Restrict to capacities \(B\) of the form
\[
B(S)=\prod_{j\in S^c}(1-p_j),
\qquad 0<p_j<1,
\]
with \(B(E)=1\) and
\[
B(\varnothing)=Q:=\prod_{j\in E}(1-p_j)\in(0,1).
\]
Their M\"{o}bius transforms are given by
\[
m_B(S)
=
\prod_{j\in S}p_j
\prod_{j\in S^c}(1-p_j)
=
Q\prod_{j\in S}\frac{p_j}{1-p_j},
\qquad S\subseteq E,
\]
where the empty-set case gives
\(m_B(\varnothing)=Q\).

On this product-structured class, the correspondence
\[
T\longleftrightarrow\widehat{T}
\longleftrightarrow V\longleftrightarrow v_P
\]
is bijective. The inverse map recovers the parameters from
\[
p_j=
\frac{m_B(\{j\})}
{Q+m_B(\{j\})},
\qquad j\in E,
\]
and then reconstructs \(T\), \(\widehat{T}\), \(V\), and
\(v_P\) through the displayed relations.
\end{enumerate}
\end{theorem}

\begin{proof}
\textbf{(a) Forward direction.}

Let
\[
Q=\prod_{j\in E}(1-p_j),
\qquad
c=1-Q,
\qquad
a_j=\frac{1}{1-p_j}.
\]
Since \(0<p_j<1\) for every \(j\in E\), we have
\(0<Q<1\) and \(c>0\).

By definition,
\[
\widehat{T}(S)
=
\frac{T(S)}{c}
=
\frac{1-\prod_{j\in S}(1-p_j)}{1-Q}.
\]
Because \(T\) is monotone and \(T(\varnothing)=0\),
\(T(E)=c\), the function \(\widehat{T}\) is a normalized
monotone capacity satisfying
\[
\widehat{T}(\varnothing)=0,
\qquad
\widehat{T}(E)=1.
\]
However, \(\widehat{T}\) is not, in general, a belief function.
Indeed, for every nonempty \(S\subseteq E\), its M\"{o}bius
transform is
\[
m_{\widehat{T}}(S)
=
\frac{(-1)^{|S|+1}}{c}
\prod_{j\in S}p_j.
\]
Consequently, if \(n\geq 2\), every two-element coalition
\(\{i,j\}\) has
\[
m_{\widehat{T}}(\{i,j\})
=
-\frac{p_i p_j}{c}<0.
\]
Thus \(\widehat{T}\) is not a belief function when
\(n\geq 2\).

Next, since \(S\sqcup S^c=E\),
\[
B(S)
=
\prod_{j\in S^c}(1-p_j)
=
Q\prod_{j\in S}a_j.
\]
In particular,
\[
B(\varnothing)=Q,
\qquad
B(E)=1.
\]

The Möbius transform of \(B\) is
\[
\begin{aligned}
m_B(S)
&=
\sum_{A\subseteq S}
(-1)^{|S|-|A|}B(A)\\
&=
\sum_{A\subseteq S}
(-1)^{|S|-|A|}
\prod_{j\in A^c}(1-p_j)\\
&=
\prod_{j\in S^c}(1-p_j)
\prod_{j\in S}\bigl(1-(1-p_j)\bigr)\\
&=
\prod_{j\in S^c}(1-p_j)
\prod_{j\in S}p_j\\
&=
Q\prod_{j\in S}\frac{p_j}{1-p_j},
\qquad S\subseteq E.
\end{aligned}
\]
The formula also holds for \(S=\varnothing\), giving
\(m_B(\varnothing)=Q\). Since every Möbius coefficient
is nonnegative, \(B\) is an unnormalized belief-side
capacity. Because \(B(\varnothing)=Q>0\), it is not a
normalized belief function.

Define
\[
V(S)=1-\widehat{T}(S^c).
\]
Since
\[
\widehat{T}(S^c)
=
\frac{1-B(S)}{c},
\]
we obtain
\[
\begin{aligned}
V(S)
&=
1-\frac{1-B(S)}{c}\\
&=
\frac{B(S)-Q}{c}\\
&=
\frac{Q}{c}
\left(\prod_{j\in S}a_j-1\right)\\
&=
\frac{Q}{c}\,v_P(S).
\end{aligned}
\]
Here
\[
v_P(S)=\prod_{j\in S}a_j-1.
\]

For the Möbius transform of \(V\), subtracting the
constant \(Q\) changes only the empty-set coefficient.
Hence
\[
m_V(\varnothing)=0
\]
and, for every nonempty \(S\subseteq E\),
\[
\begin{aligned}
m_V(S)
=\frac{1}{c}m_B(S)
=\frac{Q}{c}
\prod_{j\in S}\frac{p_j}{1-p_j}
\geq 0.
\end{aligned}
\]
Moreover,
\[
V(\varnothing)=0,
\qquad
V(E)=1.
\]
Therefore \(V\) is a normalized belief function.

The identities
\[
B(S)=Q\bigl(v_P(S)+1\bigr),
\qquad
V(S)=\frac{Q}{c}v_P(S),
\qquad
v_P(S)=\frac{c}{Q}V(S)
\]
hold for every \(S\subseteq E\). In particular,
\[
V(\varnothing)=v_P(\varnothing)=0,
\qquad
V(E)=1,
\qquad
v_P(E)=\frac{c}{Q}.
\]
Thus the multiplicative relation between \(V\) and \(v_P\)
holds on all coalitions, including the boundary cases.

\smallskip
\noindent
\textbf{(b) Bijection on the product-structured class.}

Consider the class of unnormalized belief-side capacities
\(B\) satisfying \(B(E)=1\), with Möbius transform
\[
m_B(\varnothing)=Q
=
\prod_{j\in E}(1-p_j),
\qquad 0<p_j<1,
\]
and
\[
m_B(S)
=
Q\prod_{j\in S}\frac{p_j}{1-p_j},
\qquad \varnothing\neq S\subseteq E.
\]
The parameters \(p_j\) determine all the M\"{o}bius
coefficients, and hence determine \(B\) uniquely.

Conversely, suppose \(B\) belongs to this class and its
Möbius transform is known. Since
\[
m_B(\varnothing)=B(\varnothing)=Q,
\]
the singleton coefficients satisfy
\[
m_B(\{j\})=Q\frac{p_j}{1-p_j}.
\]
Solving for \(p_j\) gives
\[
p_j=
\frac{m_B(\{j\})}
{Q+m_B(\{j\})},
\qquad j\in E.
\]
Thus the singleton coefficients and \(Q\) recover every
parameter \(p_j\) uniquely.

To verify the corresponding capacity, Möbius inversion
gives
\[
\begin{aligned}
B(S)
&=\sum_{A\subseteq S}m_B(A)\\
&=Q\sum_{A\subseteq S}
\prod_{j\in A}\frac{p_j}{1-p_j}\\
&=Q\prod_{j\in S}
\left(1+\frac{p_j}{1-p_j}\right)\\
&=Q\prod_{j\in S}\frac{1}{1-p_j}\\
&=\prod_{j\in S^c}(1-p_j).
\end{aligned}
\]
This recovers the original product structure, with
\(B(\varnothing)=Q\) and \(B(E)=1\). In particular,
\(Q\) is determined by the recovered parameters and is
not an independent parameter.

Once the \(p_j\) have been recovered, define
\[
c=1-Q,
\qquad
T(S)=1-\prod_{j\in S}(1-p_j),
\qquad
\widehat{T}(S)=\frac{T(S)}{c},
\]
and
\[
V(S)=\frac{B(S)-Q}{c},
\qquad
v_P(S)=\frac{B(S)}{Q}-1.
\]
The identities established in part (a) then give
\[
V(S)=1-\widehat{T}(S^c)
=\frac{Q}{c}v_P(S).
\]
Consequently, every member of the product-structured
class determines a unique parameter vector and hence
a unique collection of corresponding functions
\(T,\widehat{T},V,v_P\). Conversely, every admissible
parameter vector produces exactly one such collection.

Therefore, the correspondence is bijective on the
product-structured class, with the inverse map explicitly
given by
\[
p_j=\frac{m_B(\{j\})}
{m_B(\varnothing)+m_B(\{j\})}.
\]
\end{proof}

\begin{remark}[The reverse direction fails in general]
\label{RemReverseFails}
Not every capacity \(B\) with \(B(E)=1\), and not every cooperative product game
\(v_P(S)=\prod_{j\in S}a_j-1\) (with arbitrary \(a_j\)), arises from an
inhomogeneous Bernoulli plausibility game. A counterexample is the
additive belief function of Example~\ref{ExAdditiveCounterexample},
whose second-order M\"obius masses vanish while every Bernoulli-derived
belief function has strictly positive second-order M\"obius masses.
\end{remark}

\begin{proof}[Proof of Remark~\ref{RemReverseFails}]
See Example~\ref{ExAdditiveCounterexample}.
\end{proof}

\subsubsection{Counterexample for the Reverse Direction}

\begin{example}
\label{ExAdditiveCounterexample}
Let \(E=\{1,2,3\}\) and let \(B^{\mathrm{add}}\) be the capacity with M\"obius
masses
\[
m_{B^{\mathrm{add}}}(\{1\})=m_{B^{\mathrm{add}}}(\{2\})
=m_{B^{\mathrm{add}}}(\{3\})=\tfrac13,
\qquad
m_{B^{\mathrm{add}}}(S)=0\ \text{for }|S|\ge2 .
\]
By M\"obius inversion, \(B^{\mathrm{add}}(S)=\sum_{T\subseteq S}m_{B^{\mathrm{add}}}(T)\),
so
\[
B^{\mathrm{add}}(\varnothing)=0,
\qquad
B^{\mathrm{add}}(\{i\})=\tfrac13,
\qquad
B^{\mathrm{add}}(\{i,j\})=\tfrac23,
\qquad
B^{\mathrm{add}}(E)=1 .
\]
Thus \(B^{\mathrm{add}}\) is a normalized belief function
(\(B^{\mathrm{add}}(\varnothing)=0\), \(B^{\mathrm{add}}(E)=1\)). Since all
M\"obius masses above degree \(1\) vanish, \(B^{\mathrm{add}}\) is additive.

Now consider any Bernoulli model with \(p_1,p_2,p_3\in(0,1)\). The belief-side
object it induces is
\[
V(S)=1-\widehat{T}(S^c)=\frac{B(S)-Q}{c}
=\frac{Q}{c}\,v_P(S),
\]
which is again a normalized belief function: \(V(\varnothing)=0\), \(V(E)=1\).
Its M\"obius transform is
\[
m_V(\varnothing)=0,
\qquad
m_V(S)=\frac{Q}{c}\prod_{j\in S}\frac{p_j}{1-p_j}
\quad (S\neq\varnothing).
\]
In particular, for every pair \(\{i,j\}\),
\[
m_V(\{i,j\})=\frac{Q}{c}\cdot\frac{p_ip_j}{(1-p_i)(1-p_j)}>0,
\]
because each factor is strictly positive. By contrast,
\(m_{B^{\mathrm{add}}}(\{i,j\})=0\) for every pair. Since \(V\) and
\(B^{\mathrm{add}}\) share the boundary values \(V(\varnothing)=B^{\mathrm{add}}(\varnothing)=0\)
and \(V(E)=B^{\mathrm{add}}(E)=1\) but disagree on every second-order M\"obius
mass, no choice of \(p_1,p_2,p_3\in(0,1)\) can produce
\(V=B^{\mathrm{add}}\). Hence the reverse direction fails: not every normalized
belief function arises from an inhomogeneous Bernoulli plausibility game.
The obstruction is the absence of product structure (vanishing second-order
M\"obius masses), not any failure of additivity or normalization.
\end{example}

\subsubsection{Interpretation}

The forward direction is structural: the inhomogeneous Bernoulli
plausibility game is, up to normalization and conjugacy, a belief
function and a cooperative product game. The reverse direction requires
an additional assumption: the belief function must be
\emph{product-structured}, i.e., its M\"obius transform must factorize
multiplicatively over players. Within this class the correspondence is a
bijection; outside it the correspondence fails. In particular, every
product game \(v_P\) with all \(a_j>1\) comes from a Bernoulli model, because
\(p_j=1-1/a_j\in(0,1)\) is then automatic.

\subsubsection{Shapley Values Across the Four Representations}

The Shapley value is invariant under conjugacy,
\[
\phi_i(v)=\phi_i\bigl(1-v(\cdot^{\,c})\bigr),
\]
and equivariant under positive scaling,
\(\phi_i(\alpha v)=\alpha\,\phi_i(v)\) for \(\alpha>0\). Three distinct
scale factors occur in the chain, and they must not be conflated:
\[
\phi_i(T)=c\,\phi_i(\widehat T),
\qquad
\phi_i(\widehat T)=\phi_i(V),
\qquad
\phi_i(V)=\frac{Q}{c}\,\phi_i(v_P),
\]
where the middle equality holds because \(\widehat T(S)=1-V(S^c)\), i.e.\
\(\widehat T\) and \(V\) are conjugate, and the outer two are the scaling
relations \(\widehat T=T/c\) and \(V=(Q/c)v_P\). Composing them,
\[
\phi_i(T)=c\,\phi_i(\widehat T)=c\,\phi_i(V)
=c\cdot\frac{Q}{c}\,\phi_i(v_P)=Q\,\phi_i(v_P),
\]
so the factor \(c\) cancels exactly and the net scaling from \(v_P\) to \(T\)
is \(Q\). Consequently \(\phi_i(T)\neq\phi_i(V)\) whenever \(c\neq1\)
(they differ by the multiplicative factor \(c\)), and the correct net scaling
from \(v_P\) back to \(T\) is \(Q\), not \(Q/c\).

\subsubsection{Summary Table}

The following table records the status of each object with respect to the
product-game template \(v_P(S)=\prod_{i\in S}a_i-1\), with
\(a_i=(1-p_i)^{-1}\), \(Q=\prod_{j\in E}(1-p_j)\), and \(c=1-Q\).

\begin{table}[htbp]
\centering
\renewcommand{\arraystretch}{1.3}
\caption{Which of \(T\), \(\widehat{T}\), \(V\), \(B\), \(v_P\) is a product game.}
\label{TabProductGameSummary}
\begin{tabular}{@{}p{3.4cm} p{3.4cm} p{7.2cm}@{}}
\toprule
Object & Product game? & Form and value at \(E\) \\
\midrule
\(T(S) = 1 - \prod_{j\in S}(1-p_j)\)
  & No
  & \(1 - (\text{product})\), \quad \(T(E)=c\le1\). \\
\(\widehat{T}(S) = \frac{1}{c} T(S)\)
  & No
  & \(\tfrac{1}{c}[1 - (\text{product})]\), \quad \(\widehat{T}(E)=1\). \\
\(V(S) = 1 - \widehat{T}(S^{c}) = \dfrac{B(S)-Q}{c}\)
  & Positive multiple of \(v_P\), \(V=(Q/c)\,v_P\)
  & \(\tfrac{Q}{c}\bigl[\prod_{j\in S} a_j - 1\bigr]\), \quad \(V(E)=1\). \\
\(B(S) = Q\prod_{j\in S} a_j\)
  & Positive multiple of \(v_P+1\), \(B=Q(v_P+1)\)
  & \(Q\,(v_P(S)+1)\), \quad \(B(E)=1\). \\
\(v_P(S) = \prod_{j\in S} a_j - 1\)
  & Yes, by definition
  & \(\prod_{i\in S} a_i - 1\), \quad \(v_P(E)=1/Q-1=c/Q\). \\
\bottomrule
\end{tabular}
\end{table}

The essential points are: (i) normalization is irrelevant to being a product
game; (ii) neither \(T\) nor \(\widehat T\) is a product game, for algebraic
reasons; (iii) \(V\) is the positive scalar multiple \((Q/c)\,v_P\) of \(v_P\);
(iv) \(B\) is the positive scalar multiple \(Q\,(v_P+1)\) of
\(v_P+1\); and (v) the Shapley values satisfy
\(\phi_i(\widehat T)=\phi_i(V)\) and
\(\phi_i(T)=c\,\phi_i(V)=Q\,\phi_i(v_P)\).

\subsubsection{From the Finite Shapley Sum to a One-Dimensional Integral}

The conversion from the finite Shapley sum to the one-dimensional integral
representation proceeds in four steps.

\begin{enumerate}
\item \emph{Marginal contribution.} By independence,
\[
T(S\cup\{i\})-T(S)=p_i\prod_{j\in S}(1-p_j).
\]

\item \emph{Factorial identity.} Replace the Shapley coefficient using the
beta integral,
\[
\frac{s!\,(n-1-s)!}{n!}=\int_0^1 t^{s}(1-t)^{\,n-1-s}\,dt,
\qquad s=|S|.
\]

\item \emph{Interchange sum and integral.} Since the sum over
\(S\subseteq E\setminus\{i\}\) is finite, exchange it with the integral:
\[
\phi_i(T)=p_i\int_0^1
\sum_{S\subseteq E\setminus\{i\}}
t^{|S|}(1-t)^{\,n-1-|S|}
\Bigl(\prod_{j\in S}(1-p_j)\Bigr)\,dt,
\]
where \(n-1-|S|=|E\setminus\{i\}\setminus S|\).

\item \emph{Product expansion.} Apply the binomial product identity
\[
\prod_{j\ne i}\bigl[t(1-p_j)+(1-t)\bigr]
=
\sum_{S\subseteq E\setminus\{i\}}
\Bigl(\prod_{j\in S}t(1-p_j)\Bigr)
\Bigl(\prod_{j\notin S,\,j\ne i}(1-t)\Bigr).
\]
Matching the right-hand side to the integrand of Step~3, using
\(\prod_{j\notin S,j\ne i}(1-t)=(1-t)^{\,n-1-|S|}\), gives
\[
\phi_i(T)=p_i\int_0^1\prod_{j\ne i}\bigl[t(1-p_j)+(1-t)\bigr]\,dt
=p_i\int_0^1\prod_{j\ne i}(1-p_jt)\,dt,
\]
since \(t(1-p_j)+(1-t)=1-p_jt\).
\end{enumerate}

Equivalently, the substitution \(u=1-t\) gives
\[
\int_0^1\prod_{j\ne i}(1-p_jt)\,dt
=\int_0^1\prod_{j\ne i}\bigl(1-p_j+p_ju\bigr)\,du,
\]
so
\[
\phi_i(T)=p_i\int_0^1\prod_{j\ne i}(1-p_j+p_jt)\,dt.
\]

The product-game formulation follows the same four steps with the
replacements \(p_i\to a_i-1\) and \(1-p_j\to a_j\) (note \(a_j\), not
\(a_j^{-1}\)): the marginal contribution is
\(v_P(S\cup\{i\})-v_P(S)=(a_i-1)\prod_{j\in S}a_j\), and the binomial
product identity becomes
\[
\prod_{j\ne i}\bigl[ta_j+(1-t)\bigr]
=\prod_{j\ne i}\bigl[1+t(a_j-1)\bigr],
\]
yielding
\[
\phi_i(v_P)=(a_i-1)\int_0^1\prod_{j\ne i}\bigl[1+t(a_j-1)\bigr]\,dt.
\]

\subsubsection{Computational Complexity Across the Equivalent Representations}

A natural question is whether the computational complexity differs when the
Shapley value is analyzed directly from \(T\), or via the chain
\[
T\;\longrightarrow\;\widehat T=\tfrac{1}{c}T
\;\longrightarrow\;V=1-\widehat T(\cdot^{\,c})
\;\longrightarrow\;v_P.
\]

\paragraph{Direct \(T\)-formulation.}
The plausibility function admits the one-dimensional integral representation
\[
\phi_i(T)=p_i\int_0^1\prod_{j\ne i}(1-p_j+p_jt)\,dt.
\]
For each quadrature node \(t_k\), define \(b_j(t_k)=1-p_j+p_jt_k\). Prefix
and suffix products give all \(\prod_{j\ne i}b_j(t_k)\) in \(O(n)\)
operations at that node. The integrand is a polynomial of degree at most
\(n-1\), so Gauss--Legendre quadrature with \(K=\lceil n/2\rceil\) nodes is
exact in exact arithmetic. Hence
\[
\phi_i(T)\ \text{is computable in}\ O(nK)=O(n^2).
\]
No normalization, conjugation, or reciprocal transformation is required.

\paragraph{Product-game route.}
The product-game Shapley value has the integral representation
\[
\phi_i(v_P)=(a_i-1)\int_0^1\prod_{j\ne i}\bigl[1+t(a_j-1)\bigr]\,dt,
\qquad a_j=\frac{1}{1-p_j}.
\]
The same quadrature argument gives \(O(n^2)\) arithmetic operations.
Recovering the original Shapley value requires computing
\(Q=\prod_j(1-p_j)\) and multiplying every component by \(Q\) (since
\(\phi_i(T)=Q\,\phi_i(v_P)\)), which costs \(O(n)\) additional operations.
Thus
\[
\phi_i(T)\ \text{is computable via }v_P\text{ in }O(n^2)+O(n)=O(n^2).
\]
From the arithmetic-operation viewpoint, there is neither an asymptotic
advantage nor an asymptotic disadvantage in moving from \(T\) to \(v_P\).

\paragraph{Where the two routes differ: bit complexity.}
Although the arithmetic-operation counts agree, the sizes of the intermediate
quantities can differ. In the direct \(T\)-formulation each factor
\(1-p_j+p_jt\) lies in \([1-p_j,1]\subseteq[0,1]\), so intermediate products
stay bounded by \(1\). In the product-game formulation the weights
\(a_j=(1-p_j)^{-1}\) become large when some \(p_j\) is close to \(1\), and the
global factor \(Q\) must also be computed to convert back to \(T\). Hence
equal arithmetic complexity does not imply equal bit complexity or numerical
stability: in exact rational arithmetic, the numerators and denominators in
the product-game route can grow faster than in the direct route precisely
when some \(p_j\approx1\).

\paragraph{The role of the normalization constant \(c\).}
The constant \(c=1-Q\) is not needed to recover \(\phi_i(T)\) from the product
game: one can form the weights \(a_j\), compute \(\phi_i(v_P)\), compute
\(Q\), and obtain \(\phi_i(T)=Q\,\phi_i(v_P)\). The chain
\(T\to\widehat T\to V\to v_P\) is mathematically useful for establishing the
product-game connection, but computationally it can be bypassed, avoiding
\(c\) altogether.

\paragraph{Relation to the faster algorithms.}
The \(O(n^2)\) comparison above concerns only the representation one starts
from. The \(O(n\log^2 n)\) fast-polynomial refinement of Section~\ref{SecFFTMultipoint}
and the \(O(n)\) anchor-and-propagate scheme of Section~\ref{SecAnchor} apply
to the \(T\)-formulation directly and are unaffected by the choice between
\(T\) and \(v_P\).

\section{M\"obius Transform of Cooperative Product Games}
\label{SecMobius}

\subsection{Definition and Setup}
\label{SecMobiusSetup}

Let \(E=\{1,\dots,n\}\) and let the product game be
\[
v_P(S)=\prod_{i\in S}a_i-1,\qquad S\subseteq E,
\]
with \(a_i > 1\) (to be explained shortly).

The M\"obius transform of \(v_P\) is the unique set function \(m_P: \mathcal {F}\to\mathbb{R}\) satisfying
\[
v_P(S)=\sum_{T\subseteq S}m_P(T),\qquad \forall S\subseteq E.
\]
Equivalently, by M\"obius inversion,
\[
m_P(S)=\sum_{T\subseteq S}(-1)^{|S|-|T|}v_P(T).
\]

The M\"obius transform is the Harsanyi dividend of coalition \(S\): the marginal contribution of \(S\) that cannot be attributed to any proper sub-coalition.

The Inhomogeneous Bernoulli Trials (IBT) model describes players who independently activate with probabilities \(p_i\in(0,1)\). The probability that exactly the members of \(S\) are active is
\[
\prod_{i\in S}p_i\prod_{j\notin S}(1-p_j).
\]
Normalizing by the all-inactive baseline yields the relative success probability
\[
\prod_{i\in S}\frac{p_i}{1-p_i}.
\]
Setting
\[
a_i=\frac{1}{1-p_i}\quad\Longleftrightarrow\quad b_i=\frac{p_i}{1-p_i}
\]
recovers exactly
\[
v_P(S)=\prod_{i\in S}a_i-1=\prod_{i\in S}(1+b_i)-1 .
\]

Thus, we have $a_i > 1$. This condition \(a_i>1\) is equivalent to \(p_i>0\), so it is precisely the requirement that the product game admits a probabilistic belief-game representation as an IBT. It also guarantees:

\medskip

Monotonicity: \(v_P(S)\le v_P(T)\) for \(S\subseteq T\);

Positive M\"obius coefficients:  \(m_P(T)=\prod_{i\in T}b_i>0\), so the game is a positive combination of unanimity games;
$m_P(T)$ is interpreted as the joint odds of $T$ in the IBT and as the Harsanyi dividend of $T$ in the product game.

Convexity and supermodularity: inherited from total positivity;

Nonnegativity: \(v_P(S)\ge 0\), with \(v_P(\emptyset)=0\).

\medskip

When $a_i > 1$ for all $i$, $v_P$ is strictly superadditive (indeed convex).
If some \(a_i<1\) (i.e., \(b_i<0\)), superadditivity can fail. For two players,
\(v_P(\{1,2\})-v_P(\{1\})-v_P(\{2\})=(a_1-1)(a_2-1)=b_1b_2\), so superadditivity fails
exactly when \(b_1\) and \(b_2\) have opposite signs.
For example, with \(a_1 = 0.5\) and \(a_2=2\), we get
\(v_P(\{1\}) = -0.5\), \(v_P(\{2\}) = 1\),
\(v_P(\{1,2\}) = 0\), and
\[
0 < 0.5 = v_P(\{1\}) + v_P(\{2\}),
\]
so superadditivity is violated.

\subsection{M\"obius Transform of a Product Game}
\label{SecMobiusThm}

\begin{theorem}[M\"obius transform of a product game]
\label{ThmMobius}
For \(v_P(S)=\prod_{i\in S}a_i-1\) with \(b_i:=a_i-1\),
\[
m_P(S)=\prod_{i\in S}b_i\quad(|S|\ge1),\qquad m_P(\varnothing)=0.
\]
\end{theorem}

\begin{proof}
By M\"obius inversion,
\[
m_P(S)=\sum_{T\subseteq S}(-1)^{|S|-|T|}\Bigl(\prod_{i\in T}a_i-1\Bigr)
=\prod_{i\in S}(a_i-1)-\mathbb{1}\{S=\varnothing\},
\]
where the second term is the binomial sum \((1-1)^{|S|}\), which vanishes
for \(S\neq\varnothing\) and equals \(1\) at \(S=\varnothing\) (giving
\(m_P(\varnothing)=v_P(\varnothing)=0\)).
\end{proof}

Three properties of \(m_P\) matter for the rest of the paper.

\paragraph{Factorization.} Since \(m_P(S)=\prod_{i\in S}b_i\), the M\"obius
transform is determined by the \(n\) scalars \(b_1,\dots,b_n\) together with
the fixed empty-set value \(m_P(\varnothing)=0\), and hence not by \(2^n\)
independent values. This is the structural reason product games are
tractable.

\paragraph{Convexity.} If \(a_i>1\) for all \(i\), then \(b_i>0\) and
\(m_P(S)>0\) for every nonempty \(S\). A game with nonnegative Harsanyi
dividends is convex (supermodular):
\[
v_P(S\cup T)+v_P(S\cap T)\ge v_P(S)+v_P(T).
\]
If some \(a_i<1\), then \(b_i<0\) and convexity may fail.

\paragraph{Shapley value.} The standard dividend formula
\[
\phi_i(v)=\sum_{S\subseteq N\setminus\{i\}}\frac{1}{|S|+1}\,m(S\cup\{i\})
\]
becomes, after substituting the factorization and grouping by \(|S|=k\),
\[
\phi_i=b_i\sum_{k=0}^{n-1}\frac{e_k(b_{-i})}{k+1}
=b_i\int_0^1\prod_{j\neq i}(1+b_jt)\,dt,
\]
where \(e_k(b_{-i})\) is the \(k\)-th elementary symmetric polynomial in
\(\{b_j\}_{j\neq i}\). The factorization thus converts the Shapley sum into a
symmetric-polynomial evaluation, which is the combinatorial counterpart of
the integral representation and the gateway to the \(O(n\log^2 n)\)
algorithm.

\paragraph{Consistency with the correspondence.} The factorization passes
through the scalar-multiple correspondence \(V=(B-Q)/c=(Q/c)\,v_P\): for
nonempty \(S\),
\[
m_V(S)=\frac{Q}{c}\,m_P(S)=\frac{Q}{c}\prod_{i\in S}b_i,
\qquad
m_V(\varnothing)=V(\varnothing)=0,
\]
and for the Bernoulli plausibility and belief games,
\[
m_T(S)=(-1)^{|S|+1}\prod_{i\in S}p_i,
\qquad
m_B(S)=Q\prod_{i\in S}\frac{p_i}{1-p_i}
\quad(S\subseteq E).
\]
All four transforms are multiplicatively separable, differing only by the
scalars \(Q/c\), signs, and the substitution \(b_i=p_i/(1-p_i)\) or
\(b_i=-p_i\); this is the algebraic reason the Bernoulli network game is
robust under normalization, conjugacy, and the shift
\(B=Q(v_P+1)\).

\paragraph{Why this matters.} For a general game the M\"obius transform
requires \(2^n\) values and the Shapley value needs exponentially many coalition values in general (and is \#P-hard already for weighted voting games); for a product game
the transform factorizes into \(n\) scalars, the Shapley value is a
one-dimensional integral of a degree-\((n-1)\) polynomial, and the full
vector is computable in \(O(n^2)\) (Gauss--Legendre with product sharing) or
\(O(n\log^2 n)\) (fast multipoint evaluation). The multiplicative
factorization is the structural reason.

\section{Connecting the Inhomogeneous Bernoulli Network Game to
Cooperative Product Games}
\label{SecConnection}

This section positions the deterministic algorithms of this paper relative to the
product-game literature. The exact correspondence between the Bernoulli network game
and the cooperative product game is established in Section~\ref{SecCorrespondence},
and the integral representation in Section~\ref{SecMobiusThm}.

\subsection{Relation to exact quadrature methods}

Mohammadi et al.\ (2026) show that for a product game with \(n\) players
the Shapley value is the integral of a polynomial of degree at most
\(n-1\), so Gauss--Legendre quadrature with \(m_q=\lceil n/2\rceil\) nodes
is exact in exact arithmetic, and sharing the node products across players
(in their Proposition~4, by one division per player in log space) yields
\(O(n\,m_q)=O(n^2)\) operations for the full vector. Prefix/suffix products
reach the same order without division. The Bernoulli model is a special case:
its integrand \(\prod_{j\ne i}(1-p_j+p_jt)\) is a polynomial of degree at most
\(n-1\), so the same quadrature applies verbatim; this is the \(O(n^2)\) baseline
of Section~\ref{SecGLQuad}. The present paper's algorithmic
contributions go beyond this baseline: the fast-polynomial refinement
(\(O(n\log^2 n)\) in exact arithmetic, Section~\ref{SecFFTMultipoint}) and the
anchor-and-propagate scheme, which obtains all \(\phi_i\), \(i\ge2\), from \(\phi_1\)
in \(O(n)\) time with certified error bounds (Section~\ref{SecAnchor}).

\subsection{Summary of the connection}

This framework connects the present problem with the cooperative
product-game literature (Rosales 2014; Dehez 2024) and the
computational developments of Mohammadi et al.\ (2026).
The contribution here is to exploit the resulting product structure to
develop and analyze deterministic methods for heterogeneous Bernoulli
networks: the \(O(n^2)\) exact quadrature, the \(O(n\log^2 n)\)
fast-polynomial refinement, and the \(O(n)\) anchor-and-propagate scheme.

\[
\begin{array}{c}
\text{Inhomogeneous Bernoulli network model}\\
\downarrow\\
\text{Belief/plausibility cooperative game}\\
\downarrow\\
\text{Cooperative product-game representation}\\
\downarrow\\
\text{One-dimensional polynomial integral}\\
\downarrow\\
\text{Deterministic algorithms and quadrature methods}.
\end{array}
\]

\section{From \(\phi_1\) to the remaining Shapley values}\label{SecAnchor}

This section derives the exact relations that express \(\phi_i\), \(i\ge2\), through \(\phi_1\) and
one-dimensional integrals, and illustrates them numerically. Their approximation with certified
error bounds is given in Section~\ref{sec:error}.

\subsection{The anchor identity}

Without loss of generality, we assume
\[
0<p_1\le p_2\le\cdots\le p_n<1.
\]
Write
\[
p_i=c_i p_1,
\qquad
c_i=\frac{p_i}{p_1},
\qquad i=2,\ldots,n.
\]

Recall that
\[
\phi_i=p_i\int_0^1\prod_{j\ne i}(1-p_jt)\,dt.
\]

For \(i=1\), we have
\[
\phi_1
=p_1\int_0^1\prod_{j=2}^{n}(1-p_jt)\,dt
=p_1\int_0^1\prod_{j=2}^{n}(1-c_jp_1t)\,dt.
\]

For each \(i=2,\ldots,n\),
\[
\phi_i
=p_i\int_0^1(1-p_1t)
\prod_{\substack{j=2\\j\ne i}}^{n}(1-p_jt)\,dt
=c_i p_1\int_0^1(1-p_1t)
\prod_{\substack{j=2\\j\ne i}}^{n}(1-c_jp_1t)\,dt.
\]

Define
\[
A=\int_0^1\prod_{j=2}^{n}(1-p_jt)\,dt
=\frac{\phi_1}{p_1}, \quad
B_i=\int_0^1t
\prod_{\substack{j=2\\j\ne i}}^{n}(1-p_jt)\,dt.
\]

Since
\[
\prod_{j=2}^{n}(1-p_jt)
=(1-p_it)\prod_{\substack{j=2\\j\ne i}}^{n}(1-p_jt),
\]
we have
\[
A
=\int_0^1(1-p_it)
\prod_{\substack{j=2\\j\ne i}}^{n}(1-p_jt)\,dt
=\int_0^1
\prod_{\substack{j=2\\j\ne i}}^{n}(1-p_jt)\,dt
-p_iB_i.
\]

Therefore,
\[
\begin{aligned}
\phi_i
&=p_i\int_0^1(1-p_1t)
\prod_{\substack{j=2\\j\ne i}}^{n}(1-p_jt)\,dt
=p_i\left[
\int_0^1
\prod_{\substack{j=2\\j\ne i}}^{n}(1-p_jt)\,dt
-p_1B_i
\right]\\
&=p_i\left[A+(p_i-p_1)B_i\right],
\end{aligned}
\]
where the last equality is obtained by writing $1=(1-p_it)+p_it$ inside the integral.
Substituting
$
A=\frac{\phi_1}{p_1}$ and $
p_i=c_i p_1, $
we obtain
\[
\phi_i
=p_i\left[
\frac{\phi_1}{p_1}
+(p_i-p_1)B_i
\right]
=c_i\phi_1+c_i(c_i-1)p_1^2B_i.
\]

Consequently, the desired relationship is
\[
\phi_i
=c_i\phi_1
+c_i(c_i-1)p_1^2
\int_0^1t
\prod_{\substack{j=2\\j\ne i}}^{n}
(1-c_jp_1t)\,dt,
\qquad i=2,\ldots,n.
\]

Since \(c_i\ge1\) and the integral is positive under the assumed conditions, it holds that
\[
\phi_i\ge c_i\phi_1,
\qquad i=2,\ldots,n,
\]
where the equality holds only when $c_i = 1$, i.e. when $p_i = p_1$.


\subsection{Relations among the integrals \(B_i\), \(I_i\) and \(J_i\)}

Throughout we use the substitution \(p_j=c_jp_1\) for \(j\ge2\), and write
\[
B_i=\int_0^1 t\prod_{\substack{j=2\\j\ne i}}^{n}(1-c_jp_1t)\,dt,
\qquad
I_i=\int_0^1\prod_{\substack{j=1\\j\ne i}}^{n}(1-p_jt)\,dt,
\qquad
J_i=\int_0^1\prod_{\substack{j=2\\j\ne i}}^{n}(1-p_jt)\,dt .
\]
Since \(p_j=c_jp_1\), the integrals \(J_i\) and \(I_i\) may be written as
\[
J_i=\int_0^1\prod_{\substack{j=2\\j\ne i}}^{n}(1-c_jp_1t)\,dt,
\qquad
I_i=\int_0^1(1-p_1t)\prod_{\substack{j=2\\j\ne i}}^{n}(1-c_jp_1t)\,dt .
\]
Thus \(B_i\) and \(J_i\) share the same product, and \(I_i\) differs from
\(J_i\) by the single factor \((1-p_1t)\).

In this notation the Shapley values are simply
\[
\phi_i=p_iI_i,\qquad i=1,\ldots,n,
\]
by the formula recalled at the beginning of this section; in particular
\(I_1=A=\phi_1/p_1\). Moreover, for \(i\ge2\), \(I_1=J_i-p_iB_i\) and
\(I_i=J_i-p_1B_i\) give \(I_i=I_1+(p_i-p_1)B_i\), which is the anchor
identity \(\phi_i=p_i\bigl[A+(p_i-p_1)B_i\bigr]\) restated.
\medskip

\paragraph{1. Relation between \(B_i\) and \(J_i\).}
Write
\[
g_i(t):=\prod_{\substack{j=2\\j\ne i}}^{n}(1-c_jp_1t)
=\prod_{\substack{j=2\\j\ne i}}^{n}(1-p_jt),
\qquad
B_i=\int_0^1 t\,g_i(t)\,dt,
\qquad
J_i=\int_0^1 g_i(t)\,dt ,
\]
so \(B_i\) is the first moment of the \(J_i\)-integrand. Since
\(0<p_j=c_jp_1<1\), each factor \(1-p_jt\) lies in \((0,1]\) and is
decreasing on \([0,1]\). Hence \(g_i\) is positive and nonincreasing, with
\(g_i(0)=1\), and it is strictly decreasing unless the product is empty
(in which case \(g_i\equiv1\)). Then
\[
\frac12\prod_{\substack{j=2\\j\ne i}}^n(1-p_j)
\ \le\ B_i\ \le\ \frac{J_i}{2}.
\]

\emph{Proof.} The lower bound follows from \(g_i(t)\ge g_i(1)\) on \([0,1]\),
which gives \(B_i\ge g_i(1)\int_0^1t\,dt=\tfrac12 g_i(1)\). For the upper bound,
\(t-\tfrac12\) and \(g_i(t)-g_i(\tfrac12)\) have opposite signs on \([0,1]\), so
\((t-\tfrac12)\,g_i(t)\le(t-\tfrac12)\,g_i(\tfrac12)\) and
\[
B_i-\tfrac12J_i=\int_0^1\bigl(t-\tfrac12\bigr)g_i(t)\,dt
\le g_i(\tfrac12)\int_0^1\bigl(t-\tfrac12\bigr)\,dt=0 .
\]
Both bounds are equalities when the product is empty
(\(B_i=J_i/2=\tfrac12\)), and both are strict otherwise. \qed

\medskip
\paragraph{2. Relation between \(B_i\) and \(I_i\).}
Since \(I_i\) differs from \(J_i\) only by the factor \((1-p_1t)\), expanding
this factor gives the exact identities
\[
I_i=\int_0^1(1-p_1t)\,g_i(t)\,dt=J_i-p_1B_i,
\qquad
B_i=\frac{J_i-I_i}{p_1}\quad(p_1>0).
\]
Thus, once \(I_i\) and \(J_i\) are known, \(B_i\) requires no separate
integration. Moreover, for \(0 < p_1 < 1\), the bound \(B_i\le\tfrac12J_i\)
gives
\[
I_i-B_i=J_i-(1+p_1)B_i\ \ge\ J_i-\tfrac{1+p_1}{2}\,J_i
=\tfrac{1-p_1}{2}\,J_i\ > 0,
\]
and \(I_i=J_i-p_1B_i < J_i\). Hence
\[
B_i < I_i < J_i .
\]

\begin{remark}
Pointwise, neither weight dominates: \(1-p_1t\ge t\) only on
\([0,\tfrac1{1+p_1}]\). The ordering \(B_i\le I_i\) holds because \(g_i\) is
nonincreasing, so its mass is concentrated at small \(t\), where the weight
\(1-p_1t\) is the larger one. For example, with \(p_1=\tfrac12\) and a single factor
\(g_i(t)=1-\tfrac12t\),
\(B_i=\tfrac13<\tfrac7{12}=I_i\).
\end{remark}

\bigskip

{3. Summary of the relationships}

The three integrals satisfy
\[
I_i=J_i-p_1B_i,\quad
B_i=\tfrac{J_i-I_i}{p_1}\ (p_1>0),\quad
\tfrac12 g_i(1)\le B_i\le\tfrac12 J_i,\quad
B_i\le I_i\le J_i .
\]

\bigskip
In terms of these integrals, the Shapley value can also be written as
\[
\phi_i=p_iI_i=p_iJ_i-p_ip_1B_i ,
\]
which is equivalent to the anchor identity, since \(A=J_i-p_iB_i\).
For computation the anchor form is preferable: it reuses \(A=\phi_1/p_1\)
for all players and needs only one integral, \(B_i\), per player, whereas
the form above needs both \(J_i\) and \(B_i\). Its value lies in bounds:
combined with \(0\le B_i\le\tfrac12J_i\), it gives
\[
\Bigl(1-\frac{p_1}{2}\Bigr)p_iJ_i\ \le\ \phi_i\ \le\ p_iJ_i .
\]
\medskip

\paragraph{An upper bound from the arithmetic--geometric mean inequality.}
With \(K_i=\{2,\dots,n\}\setminus\{i\}\), \(m=|K_i|\ge1\), \(S_i=\sum_{j\in K_i}p_j\) and \(a=S_i/m\),
\[
\prod_{j\in K_i}(1-p_jt)\le(1-at)^m ,
\qquad
B_i\le\int_0^1t(1-at)^m\,dt
=\frac{1-(1-a)^{m+1}\bigl[1+(m+1)a\bigr]}{a^2(m+1)(m+2)}\quad(a>0).
\]
This bound depends only on \(S_i\) and \(m\). Exponential brackets for \(B_i\) are treated in
Section~\ref{sec:error}.

\subsection{A second anchor identity}\label{SecSecondAnchor}

The relations of Section~\ref{SecAnchor} lead to a second exact
formula for \(\phi_i\) that uses \(J_i\) in place of \(B_i\). We keep the
assumptions and the notation \(B_i\), \(I_i\), \(J_i\) introduced there, and
recall that \(\phi_i=p_iI_i\).

\paragraph{1. An exact formula for \(\phi_i\).}
For \(i\ge2\), pulling the factor \((1-p_1t)\) out of \(I_i\), and the factor
\((1-p_it)\) out of \(\phi_1/p_1=\int_0^1\prod_{j=2}^n(1-p_jt)\,dt\), gives
\[
I_i=J_i-p_1B_i,
\qquad
\frac{\phi_1}{p_1}=J_i-p_iB_i .
\]
Hence
\[
\phi_i-\phi_1=p_i(J_i-p_1B_i)-p_1(J_i-p_iB_i)=(p_i-p_1)J_i ,
\]
that is,
\[
\phi_i=\phi_1+(p_i-p_1)J_i=\phi_1+(c_i-1)p_1J_i,
\qquad i=2,\ldots,n.
\]
This is equivalent to the first anchor identity
\(\phi_i=c_i\phi_1+c_i(c_i-1)p_1^2B_i\), as substituting
\(\phi_1=p_1(J_i-p_iB_i)\) shows. Both forms need one integral per player,
\(J_i\) or \(B_i\). They differ in how errors propagate: here an error in
\(\phi_1\) enters \(\phi_i\) with factor \(1\) rather than \(c_i\), while the
error in \(J_i\) is damped by \(p_i-p_1\). When \(n=2\) the product in \(J_i\)
is empty, \(J_2=1\), and the identity gives \(\phi_2=\phi_1+p_2-p_1\) exactly.

\begin{remark}
The same argument applies to any two players \(i\ne k\). Writing
\(G_{ik}(t)=\prod_{j\ne i,k}(1-p_jt)\),
\[
\phi_k-\phi_i=(p_k-p_i)\int_0^1G_{ik}(t)\,dt ,
\]
so \(\phi_1\le\phi_2\le\cdots\le\phi_n\), with \(\phi_i=\phi_k\) exactly when
\(p_i=p_k\).
\end{remark}

\paragraph{2. Bounds for \(\phi_i\) via \(J_i\).}
Let \(S_i=\sum_{j\ge2,\,j\ne i}p_j\). Since
\(1-p_j\le1-p_jt\le e^{-p_jt}\) on \([0,1]\),
\[
\prod_{\substack{j=2\\j\ne i}}^n(1-p_j)
\ \le\ J_i\ \le\ \frac{1-e^{-S_i}}{S_i}\ \le\ 1,
\]
where the middle expression is read as \(1\) when \(S_i=0\). Substituting
into the identity above gives
\[
\phi_1+(p_i-p_1)\prod_{\substack{j=2\\j\ne i}}^n(1-p_j)
\ \le\ \phi_i\ \le\
\phi_1+(p_i-p_1)\,\frac{1-e^{-S_i}}{S_i}
\ \le\ \phi_1+(p_i-p_1).
\]
These bounds are explicit in the \(p_j\). They complement
\(\phi_i\ge c_i\phi_1\) from Section~\ref{SecAnchor}, and in practice the
larger of the two lower bounds can be used.

\bigskip

\paragraph{3. Implications for computation.}
The identity suggests the following way to compute the Shapley vector:
\begin{enumerate}
  \item compute \(\phi_1\);
  \item compute \(J_i\) for \(i=2,\ldots,n\);
  \item set \(\phi_i=\phi_1+(p_i-p_1)J_i\).
\end{enumerate}
Each \(J_i\) is the integral of a polynomial of degree \(n-2\), so it is
computed exactly by Gauss--Legendre quadrature with
\(\lceil (n-1)/2\rceil\) nodes. Evaluating every \(J_i\) independently costs
\(O(n^2)\) per player, \(O(n^3)\) in total. Sharing one set of nodes and
forming the products \(\prod_{j\ge2,\,j\ne i}(1-p_jt_k)\) for all \(i\) with
prefix and suffix products reduces this to \(O(n^2)\) in total. This is the
same cost as computing \(\phi_i=p_iI_i\) directly, since \(J_i\) is simply
\(I_i\) with the factor \((1-p_1t)\) removed. The identity therefore gives no
asymptotic saving, and in particular no \(O(n\log^2 n)\) algorithm. Its
computational benefit lies in error propagation: an error in \(\phi_1\) is
passed on with factor \(1\), and an error in \(J_i\) is damped by
\(p_i-p_1\), which is small for players with probabilities close to \(p_1\).

Its main use is theoretical. By isolating the difference
\(\phi_i-\phi_1=(p_i-p_1)J_i\), with \(J_i>0\), it yields the ordering
\(\phi_1\le\cdots\le\phi_n\) of the Remark and the explicit bounds of
part~2, without working with \(B_i\).

\bigskip

\emph{Comparison with QuadraSHAP in practice.} In floating point the
Gauss--Legendre nodes and weights, which are irrational
(Remark~\ref{rem:nodes-rational}), are rounded to machine precision, and the
route above costs about \(3n^2\) operations with no transcendental functions,
comparable to QuadraSHAP's formula (5) (\(3n^2+n\) operations and \(n^2\)
evaluations of \(\log\) or \(\exp\)); the node cost is the same for both. In
exact arithmetic the situation differs. QuadraSHAP cannot be exact, since any
rational computation must approximate its nodes and weights, to about
\(b+O(\log n)\) bits for accuracy \(2^{-b}\). The identity
\(\phi_i=\phi_1+(p_i-p_1)J_i\), by contrast, needs no nodes: expanding
\(Q(t)=\prod_{j\ge2}(1-p_jt)=\sum_k a_kt^k\) once and dividing by
\(1-p_it\) for each \(i\) (\(b^{(i)}_k=a_k+p_ib^{(i)}_{k-1}\)) gives
\(J_i=\sum_k b^{(i)}_k/(k+1)\), so the whole vector is obtained in \(O(n^2)\)
operations and is exact for rational \(p_j\). Its numbers grow to about \(nb\)
bits, and in floating point the alternating coefficients cause cancellation,
so this exact route suits moderate \(n\), while quadrature remains the
practical choice for large \(n\).

\subsection{Numerical examples}

We illustrate the anchor formula on the game \(T(S)=1-\prod_{j\in S}(1-p_j)\), whose Shapley
values are \(\phi_k=p_k\int_0^1\prod_{j\ne k}(1-p_jt)\,dt\). With \(p_i=c_ip_1\) we use
\[
\phi_1=p_1\int_0^1\prod_{j=2}^n(1-p_jt)\,dt \quad\text{(exact)},
\qquad
\phi_i=c_i\phi_1+c_i(c_i-1)p_1^2B_i,
\qquad
B_i=\int_0^1t\prod_{\substack{j=2\\ j\ne i}}^n(1-p_jt)\,dt .
\]
In every example we take \(i=2\), compute \(\phi_1\) exactly, estimate \(B_2\), insert the
estimate into the formula for \(\phi_2\), and compare with the exact value. The anchor identity
is an algebraic identity and holds for any labelling of the players; the examples require only
\(p_2\ge p_1\) (so \(c_2\ge1\)), and some of the remaining players have smaller probabilities
than player~\(1\). Intervals below are rounded outward, so they remain certified.

With
\[
S_2=\sum_{j=3}^n p_j,\qquad
a_2=-\sum_{j=3}^n\ln(1-p_j),\qquad
Q_2=\sum_{j=3}^n p_j^2,\qquad
g(s)=G_1(s)=\frac{1-(1+s)e^{-s}}{s^2},
\]
and \(e_r\) the \(r\)-th elementary symmetric polynomial of \(\{p_j: j\ge3\}\) (so \(e_1=S_2\)),
we compare three ways of handling \(B_2\):
\begin{itemize}
\item the first-order exponential bracket \(g(a_2)\le B_2\le g(S_2)\) (note \(a_2\ge S_2\) and
\(g\) is decreasing);
\item the two-term Bonferroni bounds
\(\tfrac12-\tfrac{e_1}{3}\le B_2\le\tfrac12-\tfrac{e_1}{3}+\tfrac{e_2}{4}\);
\item the estimator \(\hat B_2=G_1(S_2)-\tfrac12Q_2G_3(S_2)\) of~\eqref{eq:Bhat}, which is the
estimator recommended in Section~\ref{sec:error}.
\end{itemize}
The first two give certified intervals; the third gives a point estimate whose error is
controlled in Section~\ref{sec:error}.

\subsubsection*{Three regimes for \(S_2\)}

\paragraph{Case 1: small \(S_2\).}
\(n=4\), \(p_1=0.1\), \(p_2=0.2\) (\(c_2=2\)), \(p_3=0.1\), \(p_4=0.05\). Then \(S_2=0.15\) and
\(e_2=0.005\).

\textbf{Exact \(\phi_1\).} The factors \(j=2,3,4\) have elementary symmetric polynomials
\(0.35\), \(0.035\) and \(0.001\), so
\[
\phi_1=0.1\Bigl(1-\tfrac{0.35}{2}+\tfrac{0.035}{3}-\tfrac{0.001}{4}\Bigr)=0.1\cdot0.836417=0.083642 .
\]

\textbf{\(B_2\).} Since only two factors remain, the Bonferroni expansion terminates and
\(B_2=\tfrac12-\tfrac{0.15}{3}+\tfrac{0.005}{4}=0.45125\) exactly; the Bonferroni bounds are
\(0.4500\le B_2\le0.45125\), the upper one being attained. The exponential bracket is
\(g(a_2)=0.45072\le B_2\le g(S_2)=0.45271\), with \(a_2=0.15665\), and \(\hat B_2=0.45132\).

\textbf{\(\phi_2\).} The coefficient is \(c_2(c_2-1)p_1^2=2\cdot1\cdot0.01=0.02\), and
\(c_2\phi_1=0.167283\). The exponential bracket gives \(\phi_2\in[0.17629,\,0.17634]\); since
\(B_2\) is known exactly here, so is
\(\phi_2=0.167283+0.02\cdot0.45125=0.176308\).

\paragraph{Case 2: moderate \(S_2\).}
\(n=6\), \(p_1=0.2\), \(p_2=0.3\) (\(c_2=1.5\)), and \(p_3=\dots=p_6=0.25\). Then \(S_2=1\) and
\(e_2=6\cdot0.0625=0.375\).

\textbf{Exact \(\phi_1\).} \(\phi_1=0.2\int_0^1(1-0.3t)(1-0.25t)^4\,dt\). The product expands to
\(1-1.3t+0.675t^2-0.175t^3+0.02265625t^4-0.001171875t^5\); integrating gives \(0.535586\), so
\(\phi_1=0.107117\).

\textbf{\(B_2\).} The exponential bracket is \(g(a_2)=0.24127\le B_2\le g(1)=0.26425\), with
\(a_2=1.15073\). The Bonferroni bounds, \(0.1667\le B_2\le0.2604\), are much wider because \(e_2\)
is not small. The estimator gives \(\hat B_2=0.25000\); the exact value is \(B_2=0.248568\).

\textbf{\(\phi_2\).} The coefficient is \(1.5\cdot0.5\cdot0.04=0.03\), and
\(c_2\phi_1=0.160676\). The bracket gives \(\phi_2\in[0.16791,\,0.16861]\); the exact value is
\(\phi_2=0.160676+0.03\cdot0.248568=0.168133\).

\paragraph{Case 3: large \(S_2\).}
\(n=10\), \(p_1=0.5\), and \(p_j=0.6\) (\(c_j=1.2\)) for all \(j\ge2\). Then \(S_2=8\cdot0.6=4.8\).

\textbf{Exact \(\phi_1\).} All nine factors are equal, so the integral is elementary:
\[
\phi_1=0.5\int_0^1(1-0.6t)^9\,dt=0.5\cdot\frac{1-0.4^{10}}{10\cdot0.6}=0.083325 .
\]

\textbf{\(B_2\).} The exponential bracket is \(g(a_2)=0.01850\le B_2\le g(S_2)=0.04134\), with
\(a_2=-8\ln0.4=7.33\). The Bonferroni bounds are vacuous (\(e_1/3=1.6>\tfrac12\)). The estimator
gives \(\hat B_2=0.02984\); the exact value is \(B_2=0.030812\).

\textbf{\(\phi_2\).} The coefficient is \(1.2\cdot0.2\cdot0.25=0.06\), and
\(c_2\phi_1=0.099990\). The bracket gives \(\phi_2\in[0.10110,\,0.10247]\); the exact value is
\(\phi_2=0.099990+0.06\cdot0.030812=0.101838\). Here the correction term is small: \(c_2\phi_1\)
alone is within \(2\%\) of \(\phi_2\), so the loose bracket on \(B_2\) barely affects \(\phi_2\).

\begin{table}[htbp]
\centering
\caption{The three regimes for \(S_2\) (\(i=2\)).}
\label{tab:examples-small}
\small
\begin{tabular}{@{}lccccc@{}}
\toprule
Case & \(n\) & \(S_2\) & Exact \(\phi_1\) & Bracket on \(\phi_2\) & Exact \(\phi_2\) \\
\midrule
1 (small)    & 4  & 0.15 & 0.083642 & \([0.17629,\,0.17634]\) & 0.176308 \\
2 (moderate) & 6  & 1.00 & 0.107117 & \([0.16791,\,0.16861]\) & 0.168133 \\
3 (large)    & 10 & 4.80 & 0.083325 & \([0.10110,\,0.10247]\) & 0.101838 \\
\bottomrule
\end{tabular}
\end{table}

\subsubsection*{Larger \(n\)}

The next three examples repeat the three regimes with \(n=20\): player~\(1\), player~\(2\) with
\(c_2=2\) or \(1.5\), and \(18\) further players with a common probability \(q\). Hand evaluation
is no longer practical, so all values in Table~\ref{tab:examples-large} were computed by adaptive
numerical integration to about \(10^{-13}\). The last block of the table gives the same error
measures for Cases 1--3 above.

\begin{table}[htbp]
\centering
\caption{Examples with \(n=20\): \(p_1\), \(p_2=c_2p_1\), and \(p_3=\dots=p_{20}=q\).
Relative errors of \(\hat B_2\) and of \(\hat\phi_2=c_2\phi_1+c_2(c_2-1)p_1^2\hat B_2\) are
with respect to the exact values.}
\label{tab:examples-large}
\small
\setlength{\tabcolsep}{4pt}
\begin{tabular}{@{}lccc@{}}
\toprule
 & Case 1\('\) (small) & Case 2\('\) (moderate) & Case 3\('\) (large) \\
\midrule
\((p_1,\,p_2,\,q)\) & \((0.02,\,0.04,\,0.01)\) & \((0.05,\,0.075,\,0.025)\) & \((0.2,\,0.3,\,0.15)\) \\
\(S_2\) & 0.18 & 0.45 & 2.70 \\
Exact \(\phi_1\) & 0.0179431 & 0.0388024 & 0.0611899 \\
\midrule
Exact \(B_2\) & 0.443667 & 0.371544 & 0.096426 \\
\([g(a_2),\,g(S_2)]\) & \([0.443598,\,0.443863]\) & \([0.371176,\,0.372540]\) & \([0.092247,\,0.103065]\) \\
Bonferroni & \([0.4400,\,0.4438]\) & \([0.3500,\,0.3739]\) & vacuous \\
\(\hat B_2\) & 0.443668 & 0.371555 & 0.096528 \\
Rel.\ error of \(\hat B_2\) & \(2.2\times10^{-6}\) & \(3.0\times10^{-5}\) & \(1.1\times10^{-3}\) \\
\midrule
Exact \(\phi_2\) & 0.0362411 & 0.0589002 & 0.0946777 \\
Bracket on \(\phi_2\) & \([0.0362410,\,0.0362414]\) & \([0.0588995,\,0.0589021]\) & \([0.0945523,\,0.0948769]\) \\
Rel.\ error of \(\hat\phi_2\) & \(2.2\times10^{-8}\) & \(3.6\times10^{-7}\) & \(3.2\times10^{-5}\) \\
\midrule
\multicolumn{4}{@{}l}{\emph{Same quantities for Cases 1--3 above:}} \\
Rel.\ error of \(\hat B_2\) & \(1.5\times10^{-4}\) & \(5.8\times10^{-3}\) & \(3.1\times10^{-2}\) \\
Rel.\ error of \(\hat\phi_2\) & \(7.6\times10^{-6}\) & \(2.6\times10^{-4}\) & \(5.7\times10^{-4}\) \\
\bottomrule
\end{tabular}
\end{table}

\paragraph{Observations.}
\begin{itemize}
\item In every case the exact \(B_2\) and \(\phi_2\) lie inside the exponential bracket. Its width
is governed by \(a_2-S_2\), which is small when the individual \(p_j\) are small, so for many small
players (Cases 1\('\) and 2\('\)) the bracket nearly collapses to a point.
\item The Bonferroni bounds are competitive only for small \(S_2\); they are much wider for
moderate \(S_2\) and vacuous for large \(S_2\).
\item The estimator \(\hat B_2\) of Section~\ref{sec:error} is more accurate than either end of the
bracket in every case, with relative error between \(2\times10^{-6}\) and \(3\times10^{-2}\) in
\(B_2\), and at most \(6\times10^{-4}\) in \(\phi_2\).
\item The error in \(\phi_2\) is much smaller than the error in \(B_2\), because the correction
term \(c_2(c_2-1)p_1^2B_2\) is small compared with \(c_2\phi_1\): it is about \(1\%\)--\(5\%\) of
\(\phi_2\) in these examples.
\item The second anchor identity of Section~\ref{SecSecondAnchor} is confirmed as well: in every
case \(\phi_1+(p_2-p_1)J_2\) reproduces the exact \(\phi_2\); for instance, in Case~1,
\(J_2=0.926667\) and \(0.083642+0.1\cdot0.926667=0.176308\).
\end{itemize}


\section{Error analysis for \(\phi_i\)}\label{sec:error}

\subsection{Setting and notation}\label{sec:err-setting}

For \(2\le i\le n\) we use the exact representation
\[
\phi_i=c_i\phi_1+c_i(c_i-1)\,p_1^2\,B_i,\qquad p_i=c_ip_1,
\]
where \(\phi_1\) is computed exactly (in \(O(n^2)\) time). The only quantity to be approximated is
\[
B_i=\int_0^1 t\,P_i(t)\,dt,\qquad
P_i(t)=\prod_{j\in N_i}(1-p_jt),\qquad
N_i=\{2,\dots,n\}\setminus\{i\},\quad |N_i|=n-2 .
\]
Throughout, \(0\le p_j\le 1\). We write
\[
S_i=\sum_{j\in N_i}p_j,\qquad
Q_i=\sum_{j\in N_i}p_j^2,\qquad
R_i=\sum_{j\in N_i}p_j^3,\qquad
m_i=\max_{j\in N_i}p_j .
\]
The recommended estimators and their error bounds hold for all \(p_j\in[0,1]\), including
\(p_j=1\); only the logarithmic quantities \(a_i,\delta_i,\eta_i\) below, and the results that use
them, require \(m_i<1\). Since \(p_j^k\le m_i^{k-1}p_j\),
\begin{equation}\label{eq:QR-m}
Q_i\le m_iS_i,\qquad R_i\le m_i^2S_i .
\end{equation}
With \(e_k\) the elementary symmetric polynomials of \(\{p_j\}_{j\in N_i}\) (\(e_0=1\), \(e_k=0\) for
\(k>n-2\)),
\[
P_i(t)=\sum_{k=0}^{n-2}(-1)^ke_k\,t^k,\qquad
B_i=\sum_{k=0}^{n-2}\frac{(-1)^ke_k}{k+2},\qquad
e_1=S_i,\quad e_2=\frac{S_i^2-Q_i}{2}.
\]

\paragraph{Logarithmic quantities (when \(m_i<1\)).}
Let
\[
a_i=-\sum_{j\in N_i}\ln(1-p_j),\qquad
\delta_i=a_i-S_i=\sum_{j\in N_i}\sum_{k\ge2}\frac{p_j^k}{k},
\]
\[
\eta_i=a_i-S_i-\frac{Q_i}{2}=\sum_{j\in N_i}\sum_{k\ge3}\frac{p_j^k}{k}.
\]
Keeping the first term, or bounding \(1/k\) by \(\tfrac12\) (resp.\ \(\tfrac13\)) and summing the
geometric series with \(p_j\le m_i\),
\begin{equation}\label{eq:delta-eta}
\frac{Q_i}{2}\le\delta_i\le\frac{Q_i}{2(1-m_i)},\qquad
\frac{R_i}{3}\le\eta_i\le\frac{R_i}{3(1-m_i)} .
\end{equation}
We also put \(\kappa_i=1+\frac{m_i}{2(1-m_i)}=\frac{2-m_i}{2(1-m_i)}\), so that
\(a_i\le\kappa_iS_i\) by \eqref{eq:delta-eta} and \eqref{eq:QR-m}.

\paragraph{Auxiliary functions.}
For integers \(m\ge0\) and \(S>0\),
\[
G_m(S)=\int_0^1 t^m e^{-St}\,dt
=\frac{m!}{S^{m+1}}\Big(1-e^{-S}\sum_{r=0}^{m}\frac{S^r}{r!}\Big)
\;\le\;\min\Big\{\frac{1}{m+1},\ \frac{m!}{S^{m+1}}\Big\}.
\]
We write \(g=G_1\), so \(g(x)=\psi(x)/x^2\) with \(\psi(x)=1-(1+x)e^{-x}\), and
\(g'=-G_2\), \(G_2(x)=\chi(x)/x^3\) with \(\chi(x)=2-e^{-x}(x^2+2x+2)\). Both \(\psi\) and \(\chi\) are
increasing (\(\psi'(x)=xe^{-x}\), \(\chi'(x)=x^2e^{-x}\)) and tend to \(1\) and \(2\) exponentially fast.
Finally,
\[
E(S)=\int_{1/2}^{\infty}te^{-St}\,dt=e^{-S/2}\Big(\frac{1}{2S}+\frac{1}{S^2}\Big)
\;\le\;\frac{(S+1)e^{-S/2}}{S^2}.
\]

\paragraph{Error propagation to \(\phi_i\).}
Since \(\phi_i\) is linear in \(B_i\) and \(\phi_1\) is exact, any \(\hat B_i\) with
\(|\hat B_i-B_i|\le\varepsilon_i\) gives
\begin{equation}\label{eq:prop}
|\hat\phi_i-\phi_i|\;\le\;c_i(c_i-1)\,p_1^2\,\varepsilon_i\;=\;p_i\,(p_i-p_1)\,\varepsilon_i .
\end{equation}
The error in \(B_i\) is damped by \(p_i(p_i-p_1)\le p_i^2\le1\), and vanishes when \(p_i=p_1\).
Every bound on \(\varepsilon_i\) below therefore gives a bound on \(|\hat\phi_i-\phi_i|\) after
multiplication by \(p_i(p_i-p_1)\).

\subsection{Case 1: small \(S_i\) (Bonferroni truncation)}\label{sec:err-small}

For \(k\ge0\) define
\[
\hat B_i^{(k)}=\sum_{r=0}^{k}(-1)^r\frac{e_r}{r+2}.
\]

\paragraph{Lemma 1 (Bonferroni bracket).}
For \(x_1,\dots,x_m\in[0,1]\) let \(P=\prod_j(1-x_j)\), \(e_r=e_r(x_1,\dots,x_m)\),
\(T_k=\sum_{r=0}^k(-1)^re_r\) and \(T_{-1}=0\). Then for every \(k\ge-1\)
\[
0\;\le\;(-1)^{k+1}\big(P-T_k\big)\;\le\;e_{k+1}.
\]
That is, a partial sum errs with the sign of the first omitted term and by at most that term. (The
case \(k=-1\), i.e.\ \(0\le P\le1\), is included only for the induction.)

\emph{Proof.} Induction on \(m\). For \(m=0\), \(P=1\), \(e_0=1\), \(e_r=0\) for \(r\ge1\), and the
claim is immediate. Adding a variable \(x\in[0,1]\) gives \(P'=(1-x)P\), \(e'_r=e_r+xe_{r-1}\),
hence \(T'_k=T_k-xT_{k-1}\) and, for \(k\ge0\),
\[
(-1)^{k+1}(P'-T'_k)=\underbrace{(-1)^{k+1}(P-T_k)}_{\in[0,\,e_{k+1}]}
+x\underbrace{(-1)^{k}(P-T_{k-1})}_{\in[0,\,e_k]},
\]
which lies in \([0,\,e_{k+1}+xe_k]=[0,\,e'_{k+1}]\).
\hfill\(\square\)

\paragraph{Proposition 1 (truncation error on \(B_i\)).}
For every \(k\ge0\),
\[
0\;\le\;(-1)^{k+1}\big(B_i-\hat B_i^{(k)}\big)\;\le\;\frac{e_{k+1}}{k+3}
\;\le\;\frac{S_i^{\,k+1}}{(k+1)!\,(k+3)} .
\]

\emph{Proof.} Apply Lemma~1 with \(x_j=p_jt\), so that \(e_r(x)=e_rt^r\); multiply by \(t\ge0\) and
integrate over \([0,1]\). The last inequality is \(e_r\le S_i^r/r!\): expanding \(S_i^r\) over ordered
\(r\)-tuples, each \(r\)-subset occurs \(r!\) times. \hfill\(\square\)

For \(k=1,2,3\) the bound is \(S_i^2/8\), \(S_i^3/30\), \(S_i^4/144\), so the error is
\(O(S_i^{k+1})\) and the method is very accurate for small \(S_i\). By~\eqref{eq:prop},
\(|\hat\phi_i-\phi_i|\le p_i(p_i-p_1)\,e_{k+1}/(k+3)\).

\paragraph{Remarks.}
\begin{enumerate}
\item The \(e_r\) can be computed exactly, for all \(i\) together in \(O(nk)\), by removing \(p_i\) from
the elementary symmetric polynomials of \(\{p_j\}_{j\ge2}\) through \(e_r^{(i)}=e_r-p_i\,e_{r-1}^{(i)}\).
The bound \(e_{k+1}/(k+3)\) should then be used directly; the closed form
\(S_i^{k+1}/((k+1)!(k+3))\) is only a worst case.
\item Since the sign of the error is known, \(B_i\) lies between \(\hat B_i^{(k)}\) and
\(\hat B_i^{(k+1)}\), and their average has error at most \(e_{k+1}/\big(2(k+3)\big)\).
\item The alternating sum loses all accuracy in floating point for \(S_i\gtrsim30\), even when
evaluated in full; it is meant for small \(S_i\) only.
\end{enumerate}

\subsection{Case 2: large \(S_i\)}\label{sec:err-large}

\subsubsection{Exponential brackets}

Let \(\pi_k=\sum_{j\in N_i}p_j^k\) (\(\pi_1=S_i\), \(\pi_2=Q_i\), \(\pi_3=R_i\)). For \(t\in[0,1]\),
\begin{equation}\label{eq:logP}
-\ln P_i(t)=\sum_{k\ge1}\frac{\pi_k}{k}\,t^k ,
\end{equation}
with all terms nonnegative.

\paragraph{Lemma 2 (power-sum bracket).}
Assume \(m_i<1\). For \(K\ge0\) let \(L_K(t)=\sum_{k=1}^K\frac{\pi_k}{k}t^k\) and
\(\eta^{(K)}=a_i-\sum_{k=1}^K\frac{\pi_k}{k}\ge0\). Then for \(t\in[0,1]\)
\[
e^{-L_K(t)-\eta^{(K)}t^{K+1}}\;\le\;P_i(t)\;\le\;e^{-L_K(t)} .
\]

\emph{Proof.} By \eqref{eq:logP} the difference \(-\ln P_i(t)-L_K(t)=\sum_{k>K}\frac{\pi_k}{k}t^k\)
is nonnegative and, since \(t^k\le t^{K+1}\) for \(k>K\), at most \(\eta^{(K)}t^{K+1}\). \hfill\(\square\)

For \(K=0,1,2\) (with \(\eta^{(0)}=a_i\), \(\eta^{(1)}=\delta_i\), \(\eta^{(2)}=\eta_i\)) this gives
\begin{align*}
K=0:&\quad e^{-a_it}\le P_i(t)\le 1,\\
K=1:&\quad e^{-S_it-\delta_it^2}\le P_i(t)\le e^{-S_it},\\
K=2:&\quad e^{-S_it-\frac{Q_i}{2}t^2-\eta_it^3}\le P_i(t)\le e^{-S_it-\frac{Q_i}{2}t^2}.
\end{align*}
Combining the \(K=0\) lower bound (the chord bound) with the \(K=1\) upper bound gives the
first-order bracket
\begin{equation}\label{eq:first-bracket}
g(A_i)\;\le\;g(a_i)\;\le\;B_i\;\le\;g(S_i),\qquad A_i=\sum_{j\in N_i}\frac{p_j}{1-p_j},
\end{equation}
where \(a_i\le A_i\) follows from \(-\ln(1-p)\le p/(1-p)\) and \(g\) is decreasing.

\paragraph{Lemma 3 (size of \(B_i\)).}
\begin{enumerate}
\item[(i)] For all \(p_j\in[0,1]\): \(\;B_i\le g(S_i)\le 1/S_i^2\).
\item[(ii)] For all \(p_j\in[0,1]\): \(\;B_i\ge G_1(S_i)-Q_i\,G_3(S_i)-E(S_i)\); in particular
\(B_i\ge 1/(2S_i^2)\) whenever \(S_i\ge16\).
\item[(iii)] If \(m_i<1\): \(\;B_i\ge g(a_i)\ge\psi(S_i)/(\kappa_i^2S_i^2)\).
\end{enumerate}

\emph{Proof.} (i) \(P_i(t)\le e^{-S_it}\) and \(g(S)\le\int_0^\infty te^{-St}dt=1/S^2\).
(ii) For \(0\le x\le\frac12\), \(-\ln(1-x)\le x+x^2\), so on \([0,\frac12]\)
\(P_i(t)\ge e^{-S_it-Q_it^2}\ge e^{-S_it}(1-Q_it^2)\); on \([\frac12,1]\) use \(P_i\ge0\). Integrating,
\(B_i\ge G_1(S_i)-Q_iG_3(S_i)-E(S_i)\). With \(G_3(S)\le6/S^4\) and \(Q_i\le S_i\),
\(S_i^2B_i\ge\psi(S_i)-6/S_i-e^{-S_i/2}(S_i/2+1)\), which is increasing in \(S_i\) and exceeds
\(0.6\) at \(S_i=16\).
(iii) \(a_i\le\kappa_iS_i\) and \(\psi\) is increasing, so
\(g(a_i)=\psi(a_i)/a_i^2\ge\psi(S_i)/(\kappa_iS_i)^2\). \hfill\(\square\)

Thus \(B_i\asymp1/S_i^2\) for large \(S_i\), and relative errors are measured against \(1/S_i^2\).

\subsubsection{The first-order estimators}

The bracket \eqref{eq:first-bracket} suggests the estimators \(g(A_i)\), \(g(a_i)\) and \(g(S_i)\)
(here \(m_i<1\)). Its width only gives a relative error of order \(m_i\), but the endpoints behave
very differently. Let \(\bar p_i=Q_i/S_i\le m_i\), the \(p\)-weighted mean of the \(p_j\).

\paragraph{Proposition 2 (the upper endpoint \(g(S_i)\)).}
\[
\frac{Q_i}{2}G_3(S_i)-\frac{Q_i^2}{8}G_5(S_i)\;\le\;g(S_i)-B_i\;\le\;\delta_i\,G_3(S_i)
\;\le\;\frac{3\,Q_i}{(1-m_i)\,S_i^4},
\]
and consequently
\[
0\;\le\;\frac{g(S_i)-B_i}{B_i}\;\le\;\frac{3\kappa_i^2}{\psi(S_i)}\cdot\frac{\bar p_i}{(1-m_i)\,S_i}.
\]
So \(g(S_i)\) always overestimates \(B_i\), with relative error \(O(m_i/S_i)\); the lower bound shows
that the leading error \(\approx3Q_i/S_i^4\) is really present.

\emph{Proof.} By Lemma~2 (\(K=1\)) and \(e^{-y}\ge1-y\), \(P_i(t)\ge e^{-S_it}(1-\delta_it^2)\);
multiply by \(t\), integrate, and use \(G_3(S)\le6/S^4\) and \eqref{eq:delta-eta}. The lower bound is
Theorem~1(i) below; the relative bound follows from Lemma~3(iii). \hfill\(\square\)

\paragraph{Proposition 3 (the lower endpoints \(g(a_i)\), \(g(A_i)\)).}
\[
\frac{B_i-g(A_i)}{B_i}\;\ge\;\frac{B_i-g(a_i)}{B_i}\;\ge\;
\bar p_i\Big(\frac{\chi(S_i)}{2\kappa_i^3}-\frac{3}{(1-m_i)\,S_i}\Big)
\]
whenever the right-hand side is nonnegative. As \(S_i\to\infty\) the bracket tends to
\(1/\kappa_i^3>0\), so the relative error of \(g(a_i)\), and a fortiori of \(g(A_i)\), stays of order
\(\bar p_i\) however large \(S_i\) is.

\emph{Proof.} Since \(G_2\) is decreasing,
\(g(S_i)-g(a_i)\ge\delta_iG_2(a_i)\ge\frac{Q_i}{2}\cdot\frac{\chi(S_i)}{\kappa_i^3S_i^3}\), using
\(\delta_i\ge Q_i/2\), \(S_i\le a_i\le\kappa_iS_i\) and \(\chi\) increasing. Subtract the bound of
Proposition~2 on \(g(S_i)-B_i\) and divide by \(B_i\le1/S_i^2\). \hfill\(\square\)

Hence \(g(a_i)\) and \(g(A_i)\) should not be used for large \(S_i\), and \(g(S_i)\) is acceptable but
has a known leading error term, which can be removed. This gives the recommended estimator.

\subsubsection{Recommended estimator}

Define
\begin{equation}\label{eq:Bhat}
\hat B_i=G_1(S_i)-\frac{Q_i}{2}\,G_3(S_i).
\end{equation}

\paragraph{Theorem 1 (error bracket on \(B_i\)).}
\begin{enumerate}
\item[(i)] For all \(p_j\in[0,1]\): \(\;B_i-\hat B_i\le\dfrac{Q_i^2}{8}\,G_5(S_i)\).
\item[(ii)] For all \(p_j\in[0,1]\): \(\;B_i-\hat B_i\ge-\dfrac{2R_i}{3}\,G_4(S_i)-E(S_i)\).
\item[(iii)] If \(m_i<1\): \(\;B_i-\hat B_i\ge-\eta_i\,G_4(S_i)\).
\end{enumerate}

\emph{Proof.} (i) For \(x\in[0,1]\), \(1-x\le e^{-x-x^2/2}\), and \(e^{-y}\le1-y+\frac{y^2}{2}\)
for \(y\ge0\); hence
\(P_i(t)\le e^{-S_it}\big(1-\frac{Q_i}{2}t^2+\frac{Q_i^2}{8}t^4\big)\). Multiply by \(t\) and integrate.

(ii) For \(t\in[0,\frac12]\), \(p_jt\le\frac12\) and
\(\sum_{k\ge3}\frac{(p_jt)^k}{k}\le\frac{(p_jt)^3}{3(1-p_jt)}\le\frac23p_j^3t^3\), so with
\(e^{-y}\ge1-y\),
\(P_i(t)\ge e^{-S_it}\big(1-\frac{Q_i}{2}t^2-\frac{2R_i}{3}t^3\big)\) on \([0,\frac12]\). On
\([\frac12,1]\) use \(P_i\ge0\). Integrating, and bounding the omitted part
\(\int_{1/2}^1te^{-S_it}(1-\frac{Q_i}{2}t^2)\,dt\) by \(E(S_i)\), gives (ii).

(iii) By Lemma~2 with \(K=2\) and \(e^{-y}\ge1-y\),
\(P_i(t)\ge e^{-S_it}\big(1-\frac{Q_i}{2}t^2-\eta_it^3\big)\) on all of \([0,1]\). \hfill\(\square\)

Bound (iii) controls the logarithmic tail on all of \([0,1]\) and so pays the factor \(1/(1-m_i)\)
through \(\eta_i\). Bound (ii) uses it only on \([0,\frac12]\), where the weight \(te^{-S_it}\) is
concentrated for large \(S_i\), and pays for \([\frac12,1]\) with the exponentially small
\(E(S_i)\); it is therefore valid even when some \(p_j=1\). Both may be used, taking the larger
lower bound.

\paragraph{Corollary 1 (explicit bounds).}
For all \(p_j\in[0,1]\), with \(\varepsilon_i=|\hat B_i-B_i|\),
\[
\varepsilon_i\;\le\;\max\Big\{\frac{16\,R_i}{S_i^5},\ \frac{15\,Q_i^2}{S_i^6}\Big\}+\frac{(S_i+1)\,e^{-S_i/2}}{S_i^2}
\;\le\;\frac{16\,m_i^2}{S_i^4}+\frac{(S_i+1)\,e^{-S_i/2}}{S_i^2},
\]
and, for \(S_i\ge16\),
\(\;\varepsilon_i/B_i\le 32\,m_i^2/S_i^2+2(S_i+1)\,e^{-S_i/2}\). Consequently
\[
|\hat\phi_i-\phi_i|\;\le\;p_i(p_i-p_1)\Big(\max\Big\{\frac{16\,R_i}{S_i^5},\ \frac{15\,Q_i^2}{S_i^6}\Big\}
+\frac{(S_i+1)\,e^{-S_i/2}}{S_i^2}\Big).
\]

\emph{Proof.} Use \(G_m(S)\le m!/S^{m+1}\), the bound on \(E\), and \eqref{eq:QR-m}; for the relative
bound use \(B_i\ge1/(2S_i^2)\) from Lemma~3(ii); the last display is~\eqref{eq:prop}. \hfill\(\square\)

The middle expression, in \(R_i\) and \(Q_i\), is the one to use in practice: the step
\(R_i\le m_i^2S_i\) is crude when a few \(p_j\) are close to \(1\) and the rest are small. Even then
there is no blow-up, since \(m_i\le1\) and the bound still decays like \(S_i^{-4}\). Compared with
\(g(S_i)\), the relative error improves from \(O(m_i/S_i)\) to \(O(m_i^2/S_i^2)\), and compared with
\(g(a_i)\) or \(g(A_i)\) from \(O(m_i)\), at essentially no extra cost and without assuming
\(m_i<1\).

\paragraph{Remarks.}
\begin{enumerate}
\item \emph{Midpoint.} With \(\lambda_i=\frac{2R_i}{3}G_4(S_i)+E(S_i)\) (or the smaller of this and
\(\eta_iG_4(S_i)\) when \(m_i<1\)), the midpoint
\(\tilde B_i=\hat B_i+\tfrac12\big(\tfrac{Q_i^2}{8}G_5(S_i)-\lambda_i\big)\) has error at most half
the width of the interval of Theorem~1.
\item \emph{Asymptotic form.}
\(\hat B_i=\frac{1}{S_i^2}-\frac{3Q_i}{S_i^4}+O(e^{-S_i})\); the exact \(G_m\) should be kept for
moderate \(S_i\). For small \(S\), \(G_m\) should be evaluated by its series
\(G_m(S)=\sum_{r\ge0}\frac{(-S)^r}{r!\,(m+r+1)}\) to avoid cancellation.
\item \emph{Higher order.} Matching \(R_i\) as well, and splitting at \(t=\frac12\) as in
Theorem~1(ii), gives relative error \(O(m_i^3/S_i^3)\), again without assuming \(m_i<1\).
\item \emph{Use of \(\eta_i\).} If (iii) is used, \(\eta_i\) is best replaced by its upper bound
\(R_i/(3(1-m_i))\), which avoids cancellation in \(a_i-S_i-Q_i/2\); the tables below do so.
\end{enumerate}

\paragraph{Numerical illustration.}
Relative errors \(|\text{estimate}-B_i|/B_i\) for one random draw of \(p_j\) per row (\(B_i\) computed
to full precision by Gauss--Legendre quadrature); the last column is the certified relative bound
\(\max\{\lambda_i,\ Q_i^2G_5(S_i)/8\}/B_i\) of Theorem~1, with
\(\lambda_i=\min\{\frac{2R_i}{3}G_4(S_i)+E(S_i),\ \eta_iG_4(S_i)\}\).
\[
\begin{array}{rr|cccc|c}
S_i & m_i & g(A_i) & g(a_i) & g(S_i) & \hat B_i & \text{bound}\\\hline
4.2 & 0.29 & 2.5\mathrm{e}{-1} & 8.7\mathrm{e}{-2} & 9.7\mathrm{e}{-2} & 1.6\mathrm{e}{-4} & 1.1\mathrm{e}{-2}\\
8.9 & 0.29 & 3.2\mathrm{e}{-1} & 1.4\mathrm{e}{-1} & 6.8\mathrm{e}{-2} & 2.6\mathrm{e}{-3} & 7.2\mathrm{e}{-3}\\
28.4 & 0.30 & 3.6\mathrm{e}{-1} & 1.8\mathrm{e}{-1} & 2.1\mathrm{e}{-2} & 2.9\mathrm{e}{-4} & 7.6\mathrm{e}{-4}\\
89.3 & 0.60 & 7.3\mathrm{e}{-1} & 4.5\mathrm{e}{-1} & 1.6\mathrm{e}{-2} & 1.9\mathrm{e}{-4} & 4.2\mathrm{e}{-4}\\
256.8 & 0.50 & 5.9\mathrm{e}{-1} & 3.4\mathrm{e}{-1} & 4.0\mathrm{e}{-3} & 1.1\mathrm{e}{-5} & 2.7\mathrm{e}{-5}\\
1497.1 & 0.70 & 8.0\mathrm{e}{-1} & 5.2\mathrm{e}{-1} & 1.1\mathrm{e}{-3} & 8.2\mathrm{e}{-7} & 1.9\mathrm{e}{-6}
\end{array}
\]
In every case \(B_i\) lies in the interval of Theorem~1. The next table uses the same kind of draws
(most \(p_j\) uniform on \([0,0.3]\)) with one or more \(p_j\) close or equal to \(1\). ``(iii)'' and
``(ii)'' are the bounds of Theorem~1, the next column is the simplified bound
\(15m_i^2/((1-m_i)S_i^4)\) that would follow from (iii), and the last is Corollary~1; all entries are
relative to \(B_i\).
\[
\begin{array}{rl|c|cc|cc}
S_i & \text{large } p_j & \text{actual} & \text{(iii)} & \text{(ii)} &
\frac{15m_i^2}{(1-m_i)S_i^4} & \text{Cor.~1}\\\hline
9.7 & 0.99 & 1.6\mathrm{e}{-3} & 9.2\mathrm{e}{-2} & 7.6\mathrm{e}{-2} & 1.7\mathrm{e}{+1} & 2.8\mathrm{e}{-1}\\
10.4 & 0.999,\ 0.9999 & 1.5\mathrm{e}{-3} & 3.1\mathrm{e}{-1} & 7.4\mathrm{e}{-2} & 1.5\mathrm{e}{+3} & 2.3\mathrm{e}{-1}\\
31.6 & 1 & 1.5\mathrm{e}{-4} & \infty & 1.2\mathrm{e}{-3} & \infty & 1.6\mathrm{e}{-2}\\
30.6 & 1,\ 1,\ 0.999 & 6.4\mathrm{e}{-6} & \infty & 2.4\mathrm{e}{-3} & \infty & 1.8\mathrm{e}{-2}\\
154.5 & 1\ (\text{five times}) & 6.9\mathrm{e}{-6} & \infty & 5.2\mathrm{e}{-5} & \infty & 6.7\mathrm{e}{-4}
\end{array}
\]
The actual error of \(\hat B_i\) is unaffected by \(p_j\) close to \(1\); only the bounds that use the
logarithmic tail on all of \([0,1]\) degrade.

\subsection{Choosing between the two cases}\label{sec:err-choice}

Both bounds are computable for each \(i\) at negligible cost, so rather than fixing a threshold on
\(S_i\) we select the estimator with the smaller certified bound.

\begin{theorem}[Uniform certified error bound for \(B_i\)]\label{thm:main-unified}
For any \(p_j\in[0,1]\) and \(i\in\{2,\dots,n\}\), let
\[
\beta_i^{\mathrm{B}}=\frac{e_{k+1}}{k+3},\qquad
\beta_i^{\mathrm{E}}=\max\Big\{\lambda_i,\ \frac{Q_i^2}{8}G_5(S_i)\Big\},
\]
with \(\lambda_i\) as in the \emph{Midpoint} remark following Corollary~1, and let \(B_i^\star=\hat B_i^{(k)}\) if
\(\beta_i^{\mathrm{B}}\le\beta_i^{\mathrm{E}}\) and \(B_i^\star=\hat B_i\) of \eqref{eq:Bhat}
otherwise. Then
\[
|B_i^\star-B_i|\le\min\{\beta_i^{\mathrm{B}},\ \beta_i^{\mathrm{E}}\}=:\varepsilon_i,
\qquad
|\hat\phi_i-\phi_i|\le p_i\,(p_i-p_1)\,\varepsilon_i .
\]
\end{theorem}

\begin{proof}
Proposition~1 and Theorem~1 give \(|\hat B_i^{(k)}-B_i|\le\beta_i^{\mathrm{B}}\) and
\(|\hat B_i-B_i|\le\beta_i^{\mathrm{E}}\), both for all \(p_j\in[0,1]\); the rule selects the estimator
attaining the minimum, and \eqref{eq:prop} gives the bound on \(\phi_i\).
\end{proof}

For \(S_i\lesssim1\) the rule selects Bonferroni, with \(\beta_i^{\mathrm{B}}=O(S_i^{k+1})\); for
\(S_i\gg1\) it selects \(\hat B_i\), with \(\beta_i^{\mathrm{E}}=O(m_i^2/S_i^4)\). The bounds must be
evaluated with the exact \(G_m(S_i)\), not with the simplified form of Corollary~1, as the following
comparison shows.

\paragraph{Which bound in which range.}
Tables~\ref{tab:S-moderate}--\ref{tab:S-one} compare the certified bounds with the actual errors,
computed against \(B_i\) evaluated to full precision by Gauss--Legendre quadrature. ``Simplified''
is the last bound of Corollary~1; Theorem~1(ii) and (iii) use the exact \(G_m(S_i)\); ``Bonferroni''
is \(e_{k+1}/(k+3)\) with exact \(e_{k+1}\). ``\(10\times0.2\)'' means ten probabilities
\(p_j=0.2\), \(j\in N_i\), and ``\(95\times0.20\), \(p_j{=}1\)'' adds one \(p_j=1\).

\begin{table}[h]
\centering
\small
\renewcommand{\arraystretch}{1.15}
\begin{tabular}{@{}lccccccc@{}}
\toprule
Configuration & $S_i$ & $B_i$ & Actual error & Simplified & Thm~1(ii) & Thm~1(iii) & Bonferroni, $k=5$\\
\midrule
$10\times0.10$ & 1.0 & 0.258 & 2.1e-4 & 1.4 & 0.91 & 3.3e-4 & 2.6e-5\\
$10\times0.12$ & 1.2 & 0.227 & 2.9e-4 & 0.95 & 0.61 & 4.9e-4 & 7.8e-5\\
$10\times0.20$ & 2.0 & 0.137 & 4.5e-4 & 0.32 & 0.19 & 1.3e-3 & 1.7e-3\\
$10\times0.30$ & 3.0 & 0.077 & 2.1e-4 & 0.12 & 0.065 & 2.3e-3 & 1.9e-2\\
\bottomrule
\end{tabular}
\caption{Moderate $S_i$: error of $\hat B_i$ and certified bounds.}
\label{tab:S-moderate}
\end{table}

\begin{table}[h]
\centering
\footnotesize
\setlength{\tabcolsep}{3.5pt}
\renewcommand{\arraystretch}{1.15}
\begin{tabular}{@{}lcccccccc@{}}
\toprule
Configuration & $S_i$ & $B_i$ & Actual & Simplified & Thm~1(ii) & Thm~1(iii) & Bonferroni & Error of\\
 & & & error & & & & $k=5$ & $g(S_i)$\\
\midrule
$100\times0.20$ & 20 & 2.43e-3 & 1.7e-6 & 6.4e-6 & 5.2e-6 & 3.7e-6 & 9.5e3 & 7.3e-5\\
$200\times0.30$ & 60 & 2.74e-4 & 4.8e-8 & 1.1e-7 & 1.1e-7 & 1.0e-7 & 7.5e6 & 4.1e-6\\
$95\times0.20$, $p_j{=}1$ & 20 & 2.41e-3 & 7.6e-7 & 1.0e-4 & 1.0e-5 & $\infty$ & 9.3e3 & 8.9e-5\\
\bottomrule
\end{tabular}
\caption{Large $S_i$: error of $\hat B_i$ and certified bounds.}
\label{tab:S-large}
\end{table}

\begin{table}[h]
\centering
\small
\renewcommand{\arraystretch}{1.15}
\begin{tabular}{@{}lcccc@{}}
\toprule
Configuration ($S_i=1$) & $B_i$ & Error of $\hat B_i$ & Relative error & Error of $\hat B_i^{(5)}$\\
\midrule
$100\times0.01$ & 0.264 & 2.0e-6 & 0.001\% & 1.3e-4\\
$10\times0.1$ & 0.258 & 2.1e-4 & 0.1\% & 2.5e-5\\
$4\times0.25$ & 0.249 & 1.4e-3 & 0.6\% & exact\\
$2\times0.5$ & 0.229 & 6.6e-3 & 2.9\% & exact\\
$0.9$ and $10\times0.01$ & 0.189 & 2.9e-2 & 15\% & 2.8e-9\\
one $p_j=1$ & 0.167 & 4.1e-2 & 24\% & exact\\
\bottomrule
\end{tabular}
\caption{Errors at $S_i=1$ for different configurations of the $p_j$, $j\in N_i$.}
\label{tab:S-one}
\end{table}

\emph{\(S_i\approx1\).} The accuracy of \(\hat B_i\) depends on how \(S_i\) is composed
(Table~\ref{tab:S-one}). With many small \(p_j\) its error is \(0.1\%\) or less; when \(S_i\) comes
from one or a few large \(p_j\) it reaches about \(0.04\) (\(15\)--\(25\%\)), because the neglected
tail \(\sum_j\sum_{k\ge3}p_j^k/k\) is large (infinite if some \(p_j=1\)). The bounds based on
Theorem~1 also fail here: the simplified bound and Theorem~1(ii) exceed \(B_i\) itself, since
\(G_m(S)\le m!/S^{m+1}\) is crude at \(S=1\) (\(G_4(1)=0.088\) against \(24\)) and
\(E(1)=0.91\) is not small. The Bonferroni bound, in contrast, is configuration-free:
\(S_i^{k+1}/((k+1)!(k+3))=1/5760\approx1.7\times10^{-4}\) for \(k=5\) and \(2.5\times10^{-6}\) for
\(k=7\), and it is exact for \(k\ge|N_i|\). By~\eqref{eq:prop}, \(\hat B_i^{(5)}\) gives
\(|\hat\phi_i-\phi_i|\le1.7\times10^{-4}\) for every configuration with \(S_i\le1\).

\emph{Moderate \(S_i\) (about \(1\) to \(3\)).} The simplified bound and Theorem~1(ii) are still
vacuous or nearly so (Table~\ref{tab:S-moderate}). The best certificate is the smaller of
Theorem~1(iii) and Bonferroni with \(k\approx5\)--\(8\): Bonferroni is best near \(S_i=1\), and
Theorem~1(iii) takes over around \(S_i=2\) to \(3\).

\emph{Large \(S_i\).} Bonferroni fails completely (Table~\ref{tab:S-large}): its terms grow like
\(S_i^k/k!\) and far exceed \(B_i\approx1/S_i^2\). The estimator \(\hat B_i\) has relative error about
\(0.07\%\) at \(S_i=20\) and \(0.02\%\) at \(S_i=60\), \(40\) to \(90\) times smaller than that of
\(g(S_i)\), and all three certified bounds are within a factor \(2\) to \(4\) of the actual error. When
some \(p_j=1\), Theorem~1(iii) is unavailable, and Theorem~1(ii) in its \(R_i,Q_i\) form
(\(1.0\times10^{-5}\)) is ten times sharper than the simplified form, because \(m_i=1\) although only
one \(p_j\) is large.

\emph{Dependence on \(S_i\).} The term \((S+1)e^{-S/2}/S^2\) is strictly decreasing, so for fixed
\(m_i\) the simplified bound decreases like \(S_i^{-4}\). The \(R_i,Q_i\) form need not be monotone:
adding one \(p_k=1\) to \(200\) probabilities \(0.1\) raises \(S_i\) from \(20\) to \(21\) but
\(16R_i/S_i^5\) from \(1.0\times10^{-6}\) to \(4.7\times10^{-6}\), while the actual error stays of
order \(10^{-6}\).

\paragraph{Summary.}
\begin{center}
\small
\begin{tabular}{@{}lll@{}}
\toprule
Range & Best certified bound & Estimator\\
\midrule
$S_i\lesssim1$ & Bonferroni with exact $e_{k+1}$ & $\hat B_i^{(k)}$\\
$S_i\approx1$ to $3$ & $\min\{$Bonferroni ($k\approx5$--$8$), Theorem~1(iii)$\}$ & $\hat B_i^{(k)}$ or $\hat B_i$\\
$S_i\gg1$ & Theorem~1(ii) or (iii), $R_i,Q_i$ form & $\hat B_i$\\
$S_i\gg1$, some $p_j$ near $1$ & Theorem~1(ii), $R_i,Q_i$ form & $\hat B_i$\\
\bottomrule
\end{tabular}
\end{center}
In every case the error in \(\phi_i\) is the chosen bound times \(p_i(p_i-p_1)\). Even with no
information on \(B_i\) beyond \(0\le B_i\le\frac12\), the crude choice \(\hat\phi_i=c_i\phi_1\) has
absolute error at most \(\frac12p_i(p_i-p_1)\) and relative error at most
\((c_i-1)p_1^2/(2\phi_1)\).

\subsection{Computational cost}\label{sec:err-cost}

Counting each arithmetic operation (including \(\exp\) and \(\log\)) as \(O(1)\), the total cost is
\(O(n^2)\) for \(\phi_1\) plus \(O(n)\) for all \(\hat\phi_i\), \(2\le i\le n\). For the large-\(S_i\)
estimator, \(S=\sum_{j\ge2}p_j\), \(Q=\sum_{j\ge2}p_j^2\), \(R=\sum_{j\ge2}p_j^3\) and the two
largest of \(p_2,\dots,p_n\) are computed once in \(O(n)\); then \(S_i=S-p_i\), \(Q_i=Q-p_i^2\),
\(R_i=R-p_i^3\), \(m_i\) (the largest, or the second largest if the largest is \(p_i\)), the
\(G_m(S_i)\), \(\hat B_i\), its bound and \(\hat\phi_i\) cost \(O(1)\) per \(i\). The Bonferroni
estimator of order \(k\) costs \(O(nk)\) in total (Remark~1 of Section~\ref{sec:err-small}), and the
selection rule adds \(O(k)\) per \(i\). All \(B_i\) can also be computed exactly in \(O(n^2)\) by
\(\lceil n/2\rceil\)-point Gauss--Legendre quadrature, so the estimators do not change the overall
order. Their advantages are that they reduce the cost after \(\phi_1\) from \(O(n^2)\) to \(O(n)\),
and that they give closed-form, interpretable approximations, such as
\(\hat B_i\approx 1/S_i^2-3Q_i/S_i^4\), with certified error bounds.

\subsection{Limits of this analysis}\label{sec:err-limits}

The bounds are worst-case over all configurations with given \(S_i,Q_i,R_i,m_i\); for a specific
configuration the actual error is typically a few times smaller (see the tables above). The
constants in Corollary~1 come from \(G_m(S)\le m!/S^{m+1}\), the split point \(t=\frac12\) and
\(R_i\le m_i^2S_i\), and could be sharpened, but the order \(m_i^2/S_i^2\) of the relative error is
the correct one for \eqref{eq:Bhat}. The analysis assumes that \(\phi_1\) is exact; if it is computed
numerically, its error enters \(\phi_i\) with weight \(c_i\), which can exceed the damping factor
\(p_i(p_i-p_1)\) on \(\varepsilon_i\).

\subsection{\(\phi_i\) as a function of \(p_1\)}\label{sec:phi-p1}

For \(i\ge2\) let \(J_i=\int_0^1P_i(t)\,dt\), as in Section~\ref{SecAnchor}. Neither \(J_i\) nor
\(B_i\) depends on \(p_1\) or \(p_i\), and
\begin{equation}\label{eq:phi-JB}
\phi_i=p_i\int_0^1(1-p_1t)\,P_i(t)\,dt=p_i\big(J_i-p_1B_i\big).
\end{equation}

\paragraph{Form 1: \(p_i\) fixed.} \(\phi_i\) is affine in \(p_1\), with slope \(-p_iB_i\).

\paragraph{Form 2: \(c_i=p_i/p_1\) fixed.} Then
\begin{equation}\label{eq:phi-quadratic}
\phi_i=c_iJ_i\,p_1-c_iB_i\,p_1^2 ,
\end{equation}
which is the anchor identity written without \(\phi_1\). Since \(B_i\le\frac12J_i\)
(Section~\ref{SecAnchor}), the vertex \(p_1^*=J_i/(2B_i)\ge1\), so along the ray \(p_i=c_ip_1\) the
value \(\phi_i\) is increasing and concave in \(p_1\) on the whole admissible range.

Forms~1 and~2 are exact identities describing how player \(i\)'s value changes with \(p_1\).
Approximating \(J_i\approx G_0(S_i)-\frac{Q_i}{2}G_2(S_i)\) and \(B_i\approx\hat B_i\) turns them into
closed-form expressions in \(p_1\).

\paragraph{Form 3: a closed-form approximation without \(\phi_1\).}
Applying the construction of \(\hat B_i\) to the full product \(\prod_{j\ne i}(1-p_jt)\), with
\[
\sigma_i=\sum_{j\ne i}p_j=S_i+p_1,\qquad
Q_{(i)}=\sum_{j\ne i}p_j^2,\qquad
R_{(i)}=\sum_{j\ne i}p_j^3,\qquad
m_{(i)}=\max_{j\ne i}p_j,
\]
define
\begin{equation}\label{eq:phi-direct}
\tilde\phi_i=p_i\Big[G_0(\sigma_i)-\frac{Q_{(i)}}{2}\,G_2(\sigma_i)\Big],
\qquad G_0(\sigma)=\frac{1-e^{-\sigma}}{\sigma}.
\end{equation}
Asymptotically \(\tilde\phi_i=p_i\big(\sigma_i^{-1}-Q_{(i)}\sigma_i^{-3}\big)+O(e^{-\sigma_i})\): for
large \(\sigma_i\), a player's value is approximately its probability divided by the total
probability of the others. Unlike Forms~1 and~2, \eqref{eq:phi-direct} is an approximation, does not
single out player~\(1\), applies to \(i=1\) as well, and needs neither \(\phi_1\) nor any \(O(n^2)\)
computation.

\paragraph{Proposition 5 (error of \(\tilde\phi_i\)).}
For all \(p_j\in[0,1]\),
\[
-p_i\Big(\frac{2R_{(i)}}{3}G_3(\sigma_i)+\frac{e^{-\sigma_i/2}}{\sigma_i}\Big)
\;\le\;\phi_i-\tilde\phi_i\;\le\;p_i\,\frac{Q_{(i)}^2}{8}\,G_4(\sigma_i),
\]
hence
\[
|\tilde\phi_i-\phi_i|\;\le\;p_i\Big(\max\Big\{\frac{4R_{(i)}}{\sigma_i^4},\ \frac{3Q_{(i)}^2}{\sigma_i^5}\Big\}
+\frac{e^{-\sigma_i/2}}{\sigma_i}\Big)
\;\le\;p_i\Big(\frac{4\,m_{(i)}^2}{\sigma_i^3}+\frac{e^{-\sigma_i/2}}{\sigma_i}\Big),
\]
and, if \(\sigma_i\ge5\),
\(|\tilde\phi_i-\phi_i|/\phi_i\le 8\,m_{(i)}^2/\sigma_i^2+2e^{-\sigma_i/2}\).

\emph{Proof.} As for Theorem~1(i),(ii), with the weight \(t\) replaced by \(1\) and the omitted part
bounded by \(\int_{1/2}^\infty e^{-\sigma t}dt=e^{-\sigma/2}/\sigma\). The second display uses
\(G_m(\sigma)\le m!/\sigma^{m+1}\) and \(R_{(i)}\le m_{(i)}^2\sigma_i\), \(Q_{(i)}^2\le m_{(i)}^2\sigma_i^2\).
For the relative bound, as in Lemma~3(ii),
\(\sigma_i\int_0^1\prod_{j\ne i}(1-p_jt)\,dt\ge1-e^{-\sigma_i}-2/\sigma_i-e^{-\sigma_i/2}\), which
exceeds \(\frac12\) for \(\sigma_i\ge5\); so \(\phi_i\ge p_i/(2\sigma_i)\). \hfill\(\square\)

If \(m_{(i)}<1\), the lower bound can be replaced, as in Theorem~1(iii), by
\(-p_i\,\eta_{(i)}G_3(\sigma_i)\) with \(\eta_{(i)}\le R_{(i)}/(3(1-m_{(i)}))\), which removes the
exponential term that dominates for moderate \(\sigma_i\).

\paragraph{Numerical comparison.} Relative errors against the exact \(\phi_i\) (random \(p_j\); the last
row adds one \(p_j=1\)); the bound column uses the first inequality of Proposition~5 with the exact
\(G_m(\sigma_i)\):
\begin{center}
\small
\begin{tabular}{@{}rcccc@{}}
\toprule
$\sigma_i$ & $p_i(1-e^{-\sigma_i})/\sigma_i$ & $\tilde\phi_i$ & bound, Prop.~5 & exact $\phi_1$ with $\hat B_i$\\
\midrule
2 & 4.7e-2 & 2.6e-3 & 0.45 & 2.7e-4\\
8 & 2.5e-2 & 3.9e-4 & 1.9e-2 & 3.9e-5\\
32 & 6.5e-3 & 3.3e-5 & 1.9e-4 & 1.3e-6\\
82 & 5.4e-3 & 2.6e-5 & 1.3e-4 & 4.1e-7\\
23 (one $p_j=1$) & --- & 2.9e-5 & 6.9e-4 & ---\\
\bottomrule
\end{tabular}
\end{center}
The error of \(\tilde\phi_i\) is damped only by \(p_i\), while the estimator based on the exact
\(\phi_1\) has error \(p_i(p_i-p_1)|\hat B_i-B_i|\); in the table the latter is \(10\) to \(60\) times
more accurate. The closed form \eqref{eq:phi-direct} therefore suits fast screening and
interpretation, and the estimator based on \(\phi_1\) accurate computation. Like \(\hat B_i\),
\(\tilde\phi_i\) is unreliable when \(\sigma_i\approx1\) comes from a few large \(p_j\); there the
Bonferroni expansion should be used.

\section{Comparison with QuadraSHAP}\label{sec:quadrashap}

This section compares the methods of this paper with QuadraSHAP (Mohammadi et al.\ 2026), whose
main results are summarized in Section~\ref{SecScope}, item~3, in the notation introduced there.

\subsection{Cost of QuadraSHAP and relation to the present approach}\label{sec:quad-cost}

With \(m_q\) Gauss--Legendre nodes, one Shapley value costs \(O(n\,m_q)\), since at each node a
product of \(n-1\) factors is formed. In the exact regime \(m_q=\lceil n/2\rceil\) this is \(O(n^2)\)
per value and \(O(n^3)\) for the whole vector computed value by value. Their Proposition~4 shares
the work: it forms the full product \(P_q=\prod_jT_{q,j}\), \(T_{q,j}=(1-\tau_q)+\tau_qu_j\), once
per node and obtains each leave-one-out product by one division, \(P_q/T_{q,i}\), in log space to
avoid overflow and underflow. This is an algebraic identity, not a further approximation, and the
whole vector then costs \(O(n\,m_q)=O(n^2)\). The same order is reached without divisions by the
standard prefix/suffix-product variant (Section~\ref{SecGLQuad}), which we use as the exact
baseline; it is not a contribution of this paper. In parallel, the sums over \(j\) and over \(q\)
are reductions of depth \(O(\log n)\), so Proposition~4 also gives time \(O(\log n)\) with
\(O(n\,m_q/\log(n\,m_q))\) processors.

We compute one value, \(\phi_1\), exactly in \(O(n^2)\) and each further \(\phi_i\) in \(O(1)\) after
an \(O(n)\) setup, with a certified error bound that decreases as \(S_i\) grows
(Sections~\ref{SecAnchor} and~\ref{sec:error}). The whole vector thus costs \(O(n^2)\), the same
order as their exact method; the gain is that the remaining \(n-1\) values cost \(O(n)\) instead of
\(O(n^2)\), at the price of certified estimates instead of exact values.

\subsection{Model of computation, nodes and weights}\label{sec:quad-assumptions}

\begin{remark}[Assumptions behind the operation counts]\label{rem:quad-assumptions}
The cost \(O(n\,m_q)\) of Proposition~4 is an operation count in an idealized model. It rests on
four assumptions.

\begin{enumerate}
\item \emph{Unit-cost arithmetic.} Every arithmetic operation, \(\log\) and \(\exp\) counts as one,
whatever the size of its arguments. Floating-point arithmetic matches this model.

\item \emph{Precomputed nodes and weights.} The count omits the cost of \(\tau_q\) and \(\omega_q\),
which the paper says are ``computed once via the Golub--Welsch algorithm.'' The nodes are the
roots of the Legendre polynomial \(P_m\), \(m=m_q\), mapped to \([0,1]\); the Golub--Welsch method (Golub and Welsch 1969)
obtains them as the eigenvalues of the symmetric tridiagonal Jacobi matrix with off-diagonal
entries \(\beta_k=k/\sqrt{4k^2-1}\), and the weights as the squared first components of the
eigenvectors. Its cost depends on the implementation: \(O(m^2)\) with the tridiagonal QR iteration
applied to the first row only, but \(O(m^3)\) with a dense eigensolver, which is what widely used
routines such as \texttt{numpy}'s \texttt{leggauss} do. The asymptotic methods of Hale and Townsend
(2013) and Bogaert (2014) compute every node and weight independently in \(O(1)\), \(O(m)\) in
total, and nodes can also be tabulated. In the exact regime \(m=\lceil n/2\rceil\):
\begin{center}
\small
\begin{tabular}{@{}llll@{}}
\toprule
Claim & Dense solver & Golub--Welsch & \(O(m)\) method or table \\
\midrule
Sequential, \(O(n^2)\) & fails: \(O(n^3)\) & holds & holds \\
Parallel, \(O(\log n)\) & fails & fails & holds \\
\bottomrule
\end{tabular}
\end{center}
Golub--Welsch fails the parallel bound because the QR iteration is sequential. In the
approximate regime \(m_q\ll n\) the node cost is negligible. Proposition~4 is thus correct
given the nodes and weights; a complete statement should assume them precomputed or cite an
\(O(m)\) algorithm. The same applies to our quadrature for \(\phi_1\), whose \(O(n^2)\) cost is
unchanged with Golub--Welsch; the anchor step itself uses no nodes.

\item \emph{Nonzero factors.} The division \(P_q/T_{q,i}\) and \(\log|T_{q,i}|\) require
\(T_{q,i}\ne0\); the paper does not say how zero factors are handled. In our game
\(T_{q,j}=1-p_j\tau_q>0\) always.

\item \emph{Rounding and the parallel model.} Working in log space prevents overflow but not
rounding: each \(\exp(\log|P_q|-\log|T_{q,i}|)\) carries a relative error of about machine
precision times \(\bigl|\log|P_q|\bigr|\). The parallel bound assumes an idealized shared-memory
machine; with \(P\) processors the time is \(O(n\,m_q/P+\log n)\).
\end{enumerate}
\end{remark}

\begin{remark}[Rationality of the nodes and weights]\label{rem:nodes-rational}
Gauss--Legendre quadrature is exact only in exact arithmetic on its nodes and weights. The nodes
are symmetric about \(\tau=\tfrac12\), which is a node, and rational, for odd \(m\). All other
nodes are irrational for every degree checked, and for all \(m\) if Stieltjes' conjecture holds
that \(P_m\) (\(m\) even) and \(P_m(x)/x\) (\(m\) odd) are irreducible over \(\mathbb Q\)
(Baillaud and Bourget 1905; verified for particular degrees by Holt 1912, Ille 1924 and Wahab 1952, and for all degrees up to about 120 by Cullinan and Hajir 2013; a proof for all degrees has been announced by Deng 2026);
for example, they are \(\tfrac12\pm\tfrac1{2\sqrt3}\) for \(m=2\). In any case every node is
algebraic. On \([0,1]\) the weights are (Davis and Rabinowitz 1984) \(\omega_k=1/\bigl((1-x_k^2)P_m'(x_k)^2\bigr)\), rational
functions of the nodes: they are rational for \(m=2,3\) (\(\tfrac12,\tfrac12\) and
\(\tfrac5{18},\tfrac8{18},\tfrac5{18}\)), the middle weight \(1/P_m'(0)^2\) is rational for odd \(m\),
and for \(m\ge4\) all other weights are irrational under the same conjecture (a rational weight
would be shared by all conjugate nodes, whereas the weights differ between symmetric pairs), e.g.\
\((18\pm\sqrt{30})/72\) for \(m=4\).

Consequently every quadrature method with \(m\ge2\) nodes, QuadraSHAP and our quadrature for
\(\phi_1\) included, is ``exact'' only in the sense that its method error is zero and only the
rounding of \(\tau_q\), \(\omega_q\) and the arithmetic remains; exact results would require
arithmetic on algebraic numbers. With rational \(p_j\), exact rational results are obtained without
nodes by polynomial expansion: \(\prod_{j\ge2}(1-p_jt)=\sum_{k=0}^{n-1}a_kt^k\) gives
\(\phi_1=p_1\sum_ka_k/(k+1)\), and dividing by \(1-p_it\) for each \(i\) gives every \(J_i\) and
hence, by the second anchor identity \(\phi_i=\phi_1+(p_i-p_1)J_i\) (Section~\ref{SecSecondAnchor}),
the whole vector. The price is the growth of the numbers in exact arithmetic, and cancellation of
the alternating \(a_k\) in floating point for large \(n\).
\end{remark}

\subsection{Exact operation counts}\label{sec:opcounts}

Since several methods share the order \(O(n^2)\), we count operations. Additions, subtractions,
multiplications and divisions are \emph{arithmetic operations}; \(\log\), \(\exp\) and
\(\sqrt{\cdot}\) are counted separately; comparisons and sign changes are not counted. The cost of
the nodes and weights is excluded (Remark~\ref{rem:quad-assumptions}); it is the same for every
quadrature method, all of which use \(m=\lceil n/2\rceil\) nodes. Leading terms are for \(n\) even,
\(m=n/2\); all counts were checked by an operation-counting program.

\paragraph{Our \(\phi_1\).} By quadrature: per node, the \(n-1\) factors \(1-p_j\tau_q\)
(\(2(n-1)\) operations), their product (\(n-2\)) and the weight (\(1\)); then \(m-1\) additions and
the factor \(p_1\): in total \(3mn-2m=\tfrac32n^2-n\). By polynomial expansion: building
\(\sum_ka_kt^k\) costs \(\tfrac12n(n-1)\) multiplications and \(\tfrac12(n-1)(n-2)\) subtractions,
and \(p_1\sum_ka_k/(k+1)\) a further \(2n\), in total \(n^2+1\), with no nodes.

\paragraph{QuadraSHAP.} \emph{Formula (5) as written:} per node, \(1-\tau_q\), the \(n\) factors
\(T_{q,j}\) (\(2n\)), \(n\) logarithms and their sum; per pair \((i,q)\), one subtraction, one
exponential, one multiplication by \(\omega_q\) and one accumulation; per \(i\), \(u_i-1\) and a
final multiplication. In total \(6mn+n=3n^2+n\) arithmetic operations, plus \(\tfrac12n^2\)
logarithms and \(\tfrac12n^2\) exponentials. \emph{Plain division} (no log space, possible in our
game): form \(\omega_qP_q\) once per node and accumulate \(\sum_q\omega_qP_q/T_{q,i}\), at
\(5mn+m+n=\tfrac52n^2+\tfrac32n\). \emph{Formula (4), value by value:} each value costs
\(3mn-m+1\), one operation per node more than our \(\phi_1\) (for \(1-\tau_q\)), so the whole vector
costs \(n(3mn-m+1)=\tfrac32n^3-\tfrac12n^2+n\).

\paragraph{Exact baseline: prefix and suffix products.} All values use
\(\phi_i\approx p_i\sum_q\omega_q\prod_{j\ne i}b_{q,j}\), \(b_{q,j}=1-p_j\tau_q\), exact for
\(m\ge\lceil n/2\rceil\). At node \(q\) the weighted prefix products \(L_{q,1}=\omega_q\),
\(L_{q,i}=L_{q,i-1}b_{q,i-1}\), and suffix products \(R_{q,n}=1\), \(R_{q,i}=R_{q,i+1}b_{q,i+1}\),
give \(L_{q,i}R_{q,i}=\omega_q\prod_{j\ne i}b_{q,j}\). Per node this costs \(2n\) for the factors,
\(n-1\) and \(n-2\) for the prefix and suffix products, \(n-1\) for the products \(L_{q,i}R_{q,i}\),
and \(n\) additions (none at the first node); with the final \(n\) multiplications by \(p_i\),
\[
m(5n-4)+(m-1)n+n=6mn-4m=3n^2-2n ,
\]
with no divisions, logarithms or exponentials, and valid even when some factor vanishes.

\paragraph{Our anchor step, Case~2 (large \(S_i\)).} Once: \(S=\sum_{j\ge2}p_j\), the squares, \(Q\)
and \(r=\phi_1/p_1\), \(3n-4\) operations. For each \(i\), with
\[
G_1(S)=\frac{1-e^{-S}(1+S)}{S^2},\qquad
G_3(S)=\frac{6-e^{-S}\bigl(((S+3)S+6)S+6\bigr)}{S^4},
\]
\[
\hat\phi_i=r\,p_i+p_i(p_i-p_1)\,\hat B_i ,
\]
the quantities \(S_i,Q_i\) (\(2\)), \(S_i^2,S_i^4\) (\(2\)), \(G_1\) (\(4\)), \(G_3\) by Horner (\(8\)),
\(\hat B_i\) (\(3\)) and \(\hat\phi_i\) (\(5\)) cost \(24\) arithmetic operations and one exponential.
For the \(n-1\) players: \(27n-28\) operations and \(n-1\) exponentials. The certified bound adds
\(8\) operations and one square root per player.

\paragraph{Our anchor step, Case~1 (small \(S_i\)).} Once: \(e_1,\dots,e_k\) of \(\{p_j\}_{j\ge2}\)
(\(n+(k-1)(2n-1)-k^2-k\) operations), the constants \(1/(r+2)\) (\(k\)) and \(r=\phi_1/p_1\)
(\(1\)). For each \(i\): \(e_r^{(i)}\) by removing \(p_i\) (\(2k-1\)), \(\hat B_i^{(k)}\) (\(2k\)) and
\(\hat\phi_i\) (\(5\)), \(4k+4\) in total. For the \(n-1\) players,
\[
(6k+3)n-k^2-5k-2
\]
operations, with no exponentials or nodes: \(9n-8\), \(15n-16\), \(21n-26\) for \(k=1,2,3\), cheaper
than Case~2 for every \(k\le4\). The certified bound \(p_i(p_i-p_1)e_{k+1}^{(i)}/(k+3)\) adds
\(6n-2k-6\) operations in total.

\paragraph{Summary.} With \(k\) the number of Bonferroni terms and \(c\) the cost of one \(\log\) or
\(\exp\) in arithmetic operations (\(c\approx10\)--\(20\) is typical):
\begin{table}[h]
\centering
\caption{Operation counts for computing the Shapley values \(\phi_i\): our methods versus QuadraSHAP.}
\label{tab:op-counts}
\small
\setlength{\tabcolsep}{4pt}
\begin{tabular}{@{}>{\raggedright\arraybackslash}p{0.34\linewidth}lll@{}}
\toprule
Computation & Arithmetic ops & \(\log\)/\(\exp\) & Nodes/weights \\
\midrule
Our \(\phi_1\) alone, quadrature & \(\tfrac32n^2-n\) & -- & yes \\
Our \(\phi_1\) alone, polynomial & \(n^2+1\) & -- & none \\
Our whole vector, anchor (Case~1) & \(\tfrac32n^2+(6k+2)n-k^2-5k-2\) & -- & yes (via \(\phi_1\)) \\
Our whole vector, anchor (Case~2) & \(\tfrac32n^2+26n-28\) & \(n-1\) & yes (via \(\phi_1\)) \\
Whole vector, exact prefix/suffix (baseline) & \(3n^2-2n\) & -- & yes \\
Whole vector, second anchor, polynomial (\(J_i\)) & \(\approx 5n^2\) & -- & none \\
QuadraSHAP whole vector, Formula (5) & \(3n^2+n\) & \(n^2\) & yes \\
QuadraSHAP whole vector, plain division & \(\tfrac52n^2+\tfrac32n\) & -- & yes \\
QuadraSHAP whole vector, Formula (4), value by value & \(\tfrac32n^3-\tfrac12n^2+n\) & -- & yes \\
\bottomrule
\end{tabular}

\smallskip
\begin{minipage}{0.95\linewidth}\footnotesize
\emph{Notes.} The two anchor rows give \(\phi_1\) exactly and \(\phi_i\), \(i\ge2\), with certified
error bounds; all other rows give the whole vector exactly, up to the rounding of the nodes and
weights (``yes'') or exactly in rational arithmetic (``none'', Remark~\ref{rem:nodes-rational}). The
second-anchor row expands \(\prod_{j\ge2}(1-p_jt)\) once and divides by \(1-p_it\) for each \(i\);
computing the \(J_i\) by quadrature on shared nodes instead costs the same as the prefix/suffix row.
With \(\phi_1\) from the polynomial route, the anchor rows need no nodes and cost
\(n^2+(6k+3)n-k^2-5k-1\) (Case~1) and \(n^2+27n-27\) (Case~2). The anchor rows exclude the
certified bounds, which add \(6n-2k-6\) operations (Case~1) or \(8\) operations and one square root
per player (Case~2).
\end{minipage}
\end{table}

Thus QuadraSHAP by formula~(5) costs about \((3+c)n^2\), which is \(2+\tfrac23c\), about \(9\) to
\(15\), times our \(\phi_1\) alone; with plain division it costs \(\tfrac53\) times our \(\phi_1\).
Our anchor scheme costs \(\phi_1\) plus \(O(n)\), so for large \(n\) it is about \(\tfrac53\) times
cheaper than QuadraSHAP with plain division and \(2\) times cheaper than the prefix/suffix
baseline, the price being certified estimates instead of exact values for \(i\ge2\). Among the exact
methods, plain division (\(\tfrac52n^2\)) is slightly cheaper than prefix/suffix (\(3n^2\)) when a
division counts as one operation; since a division usually costs several multiplications, the two
are close in practice, and prefix/suffix needs no division and tolerates zero factors.

\subsection{The approximate regime of QuadraSHAP}\label{sec:quad-approx}

For \(m_q<\lceil n/2\rceil\) nodes, Proposition~3(ii) of Mohammadi et al.\ (2026) states that
\begin{equation}\label{eq:quad-bound}
\bigl|\widehat\phi_i-\phi_i\bigr|\le |u_i-1|\;C_\rho\,M_\rho\,\rho^{-2m_q},
\end{equation}
\[
M_\rho=\max_{z\in E_\rho}|g_i(z)|,\quad g_i(z)=\prod_{j\ne i}\bigl((1-z)+zu_j\bigr),
\]
where \(E_\rho\) is the Bernstein ellipse for \([0,1]\) with parameter \(\rho>1\). Since \(g_i\) is a
polynomial, the bound holds for every \(\rho>1\) and is useful only after minimizing over \(\rho\).
The proof gives \(C_\rho=4/(\rho^2-1)\). Trefethen (2013, Theorem~19.3), applied with \(m_q\) nodes
and mapped to \([0,1]\), gives instead \(C_\rho=\tfrac{32}{15}\,\rho^2/(\rho^2-1)\); the factor
\(\rho^2\) is missing in theirs, without which the bound tends to \(0\) as \(\rho\to\infty\) when
\(\deg g_i=2m_q\), although the quadrature error is nonzero. We use the corrected constant.

The bound guarantees geometric convergence in \(m_q\) but does not certify a computed value: the
statement gives no value or bound for \(M_\rho\), no choice of \(\rho\), and no rule for choosing
\(m_q\), which is chosen empirically; and the errors reported in their experiments (their Figure~2)
are measured against the exact vector, which the approximate regime is meant to avoid. For our game
\(M_\rho\) has a closed form.

\begin{lemma}\label{lem:M-rho}
Let \(u_j=1-p_j\) with \(p_j\in[0,1]\), so that \(g_i(z)=\prod_{j\ne i}(1-p_jz)\), and let
\(\sigma_i=\sum_{j\ne i}p_j\). Then
\[
M_\rho=\prod_{j\ne i}\Bigl(1+p_j\,\frac{(\rho-1)^2}{4\rho}\Bigr)
\le\exp\Bigl(\sigma_i\,\frac{(\rho-1)^2}{4\rho}\Bigr).
\]
\end{lemma}

\begin{proof}
\(E_\rho\) has centre \(\tfrac12\), semi-axes \(A=(\rho+\rho^{-1})/4\ge B=(\rho-\rho^{-1})/4\), and left
vertex \(z_0=\tfrac12-A=-(\rho-1)^2/(4\rho)\). For \(p_j>0\), \(|1-p_jz|=p_j|z-1/p_j|\), and
\(c=1/p_j-\tfrac12\ge0\). For \(z=\tfrac12+A\cos\theta+iB\sin\theta\),
\(|z-1/p_j|^2=\tfrac14\cos^2\theta-2Ac\cos\theta+c^2+B^2\), a convex function of
\(\cos\theta\in[-1,1]\) whose value at \(-1\) exceeds that at \(+1\) by \(4Ac\ge0\). So every factor
attains its maximum at the same point \(z_0\), where it equals \(1+p_j(\rho-1)^2/(4\rho)\); and
\(1+x\le e^x\). \hfill
\end{proof}

\begin{remark}\label{rem:M-rho-ext}
The proof only uses that each root \(1/p_j\) lies at or to the right of the centre \(\tfrac12\), so
the lemma holds for any product game with \(u_j\in[-1,1]\). If all \(u_j\notin[-1,1]\) the maximum is
at the right vertex \((\rho+1)^2/(4\rho)\); in mixed cases only the triangle-inequality bound is
available.
\end{remark}

With \(u_i-1=-p_i\), \eqref{eq:quad-bound} becomes the computable bound
\begin{equation}\label{eq:quad-bound-explicit}
\bigl|\widehat\phi_i-\phi_i\bigr|\le p_i\,K_{\rho}(\sigma_i),\qquad
K_\rho(\sigma)=\frac{32}{15}\,\frac{\rho^{2-2m_q}}{\rho^2-1}\,
\exp\Bigl(\sigma\,\frac{(\rho-1)^2}{4\rho}\Bigr),\qquad \rho>1 ,
\end{equation}
which costs about \(4\) operations and one exponential per player for a fixed \(\rho\).

\paragraph{Number of nodes (heuristic).} With \(\rho=1+\varepsilon\), \(\varepsilon\) small, the
exponent in \eqref{eq:quad-bound-explicit} is about \(\sigma_i\varepsilon^2/4-2m_q\varepsilon\),
minimized at \(\varepsilon=4m_q/\sigma_i\), and the prefactor is about \(16/(15\varepsilon)\). The
bound is then about
\[
p_i\,\frac{4\sigma_i}{15\,m_q}\,e^{-4m_q^2/\sigma_i}\qquad(m_q\ll\sigma_i/4).
\]
Since \(\phi_i\approx p_i/\sigma_i\) for large \(\sigma_i\), a relative accuracy \(\mathrm{tol}\)
requires about
\[
m_q\approx\tfrac12\sqrt{\sigma_i\ln\bigl(4\sigma_i^2/(15\,m_q\,\mathrm{tol})\bigr)}\,.
\]
For
\(\mathrm{tol}=10^{-13}\), \(\sigma_i=10^3\) gives \(m_q\approx100\) and \(\sigma_i=10^4\) gives
\(m_q\approx320\), consistent with the ``few hundred nodes'' observed by Mohammadi et al.\ (2026) for
up to \(10^4\) features. The number of nodes grows like \(\sqrt{\sigma_i}\), not like \(\sigma_i\).
This is an estimate from the bound; the actual error may be smaller.

\paragraph{Costs.} With \(m=m_q\) nodes, counted as in Section~\ref{sec:opcounts}, one value by
formula~(4) costs \(3mn-m+1\), the whole vector \(6mn+n\) by formula~(5) (plus \(mn\) logarithms and
exponentials) or \(5mn+m+n\) by plain division, and the bound \eqref{eq:quad-bound-explicit} about
\(4n\) operations and \(n\) exponentials for all players; the \(O(m^2)\) node cost is negligible
when \(m\ll n\). For example, with \(n=10^4\), \(p_j\) of mean \(0.5\) (\(\sigma\approx5000\)) and
\(\mathrm{tol}=10^{-13}\), so that \(m\approx225\):
\begin{center}
\small
\begin{tabular}{@{}ll@{}}
\toprule
Method & Arithmetic ops \\
\midrule
QuadraSHAP, approximate vector, plain division & \(5mn\approx1.1\times10^7\) \\
QuadraSHAP, exact vector, plain division & \(\tfrac52n^2=2.5\times10^8\) \\
Our anchor scheme with exact \(\phi_1\) & \(\tfrac32n^2\approx1.5\times10^8\) \\
Our anchor scheme with \(\phi_1\) by approximate quadrature & \(3mn+27n\approx7\times10^6\) \\
\bottomrule
\end{tabular}
\end{center}
In the approximate regime QuadraSHAP is thus much cheaper than our scheme with an exact \(\phi_1\).
The two approaches combine, however: \(\phi_1\) by approximate quadrature with the certified bound
\eqref{eq:quad-bound-explicit}, then \(\phi_i\), \(i\ge2\), by the anchor step, costs about \(3mn\)
against \(5mn\), and every value carries a computable certified bound. The accuracies differ,
however: the anchor step adds its own error, whose certified relative size here is
\(S_i|p_i-p_1|\varepsilon_i\approx16R_i|p_i-p_1|/S_i^4\), about \(10^{-10}\) rather than
\(10^{-13}\). The comparison is therefore fair only for tolerances of that order or larger.

\subsection{Choice of the anchor}\label{sec:anchor-choice}

\begin{remark}[Choice of the anchor and the role of an exact \(\phi_1\)]\label{rem:anchor-choice}
The estimate \(\hat\phi_i=c_i\hat\phi_1+p_i(p_i-p_1)\hat B_i\) has two sources of error:
\[
\hat\phi_i-\phi_i
=\underbrace{c_i\bigl(\hat\phi_1-\phi_1\bigr)}_{\text{from the anchor}}
+\underbrace{p_i(p_i-p_1)\bigl(\hat B_i-B_i\bigr)}_{\text{from estimating }B_i}.
\]

\emph{Error from the anchor.} If \(\phi_1\) is computed by approximate quadrature, then
\(|\hat\phi_1-\phi_1|\le p_1K_\rho(\sigma_1)\) by \eqref{eq:quad-bound-explicit}, so
\(c_i|\hat\phi_1-\phi_1|\le p_iK_\rho(\sigma_1)\). This is the bound \(p_iK_\rho(\sigma_i)\) for computing
\(\phi_i\) directly by quadrature, up to the factor \(e^{(p_i-p_1)(\rho-1)^2/(4\rho)}\approx1\): a large
\(c_i\) multiplies an error that is itself proportional to \(p_1\), so the anchor error is passed on
without amplification, whichever player is the anchor. If \(\phi_1\) is exact, this term is zero.

\emph{Error from estimating \(B_i\).} Even with \(\phi_1\) exact,
\(|\hat\phi_i-\phi_i|\le p_i|p_i-p_1|\,\varepsilon_i\), a relative error of about
\(S_i|p_i-p_1|\,\varepsilon_i\) for large \(S_i\). It vanishes for \(p_i=p_1\) and grows with
\(|p_i-p_1|\). The anchor identity holds for any labelling: with the smallest \(p_1\) the correction is
nonnegative (\(\phi_i\ge c_i\phi_1\)) and the largest \(p_i\) are least accurate; with the largest
\(p_1\) the correction is nonpositive and the smallest \(p_i\) are least accurate.

\emph{Choosing the anchor.} The cost is the same for every anchor, so it can be chosen for accuracy:
\begin{itemize}
\item a \(p_1\) near the midrange \(\tfrac12(\min_jp_j+\max_jp_j)\) minimizes \(\max_i|p_i-p_1|\);
\item a \(p_1\) at the \(p\)-weighted median minimizes \(\sum_ip_i|p_i-p_1|\);
\item if one \(p_j\) is close to \(1\), making it the anchor removes it from every \(N_i\). This
lowers \(m_i\) in the bound and, when \(S_i\approx1\), the case where \(\hat B_i\) is least reliable,
moves those players into Case~1, where Bonferroni is accurate.
\end{itemize}
With \(k\) exact anchors spread over the range of the \(p_j\), each player uses its nearest anchor,
the largest \(|p_i-p_1|\) falls to about \((\max_jp_j-\min_jp_j)/(2k)\), and the cost is
\(\tfrac32kn^2\). Finally, since a certified bound is computed for every player, those whose bound
exceeds a tolerance are known and can be computed exactly at \(\tfrac32n^2\) operations each.
\end{remark}

\bigskip

\section{Worked Examples, Applications, and Method Comparison}
\label{SecWorkedExamples}

This section illustrates the anchor identities on a small instance and an application, and then
compares the methods of this paper with permutation sampling (Castro et al.\ 2009, 2017) and with
the FFT--multipoint method, on a product game with \(n=100\) players. Throughout,
\[
T(S)=1-\prod_{j\in S}(1-p_j),\qquad
\phi_i=p_i\int_0^1\prod_{j\ne i}(1-p_jt)\,dt,\qquad
\sum_{i=1}^n\phi_i=V=1-\prod_{j=1}^{n}(1-p_j).
\]

\subsection{A three-player instance}

Take \(p=(\tfrac12,\tfrac13,\tfrac16)\). Then \(V=1-\tfrac12\cdot\tfrac23\cdot\tfrac56=\tfrac{13}{18}\), and
the exact Shapley values are
\[
\phi_1=\frac{83}{216}\approx0.3843,\qquad
\phi_2=\frac{25}{108}\approx0.2315,\qquad
\phi_3=\frac{23}{216}\approx0.1065,
\]
which sum to \(\tfrac{13}{18}\). Each integrand has degree two, so the two-node Gauss--Legendre rule
(nodes \(\tfrac12\mp\tfrac1{2\sqrt3}\), weights \(\tfrac12\)) gives any single value exactly.

\paragraph{From one value to the others.} Take the player with the smallest probability,
\(p_3=\tfrac16\), as the anchor and compute \(\phi_3=\tfrac{23}{216}\) exactly. With one factor left
in each product, the integrals of Section~\ref{SecAnchor} are elementary:
\[
B_1=\int_0^1t\bigl(1-\tfrac13t\bigr)dt=\tfrac7{18},\qquad
B_2=\int_0^1t\bigl(1-\tfrac12t\bigr)dt=\tfrac13,\qquad
J_1=\tfrac56,\qquad J_2=\tfrac34 .
\]
The anchor identity \(\phi_i=c_i\phi_3+c_i(c_i-1)p_3^2B_i\), with \(c_1=3\) and \(c_2=2\), gives
\[
\phi_1=\tfrac{69}{216}+\tfrac16\cdot\tfrac7{18}=\tfrac{83}{216},\qquad
\phi_2=\tfrac{46}{216}+\tfrac1{18}\cdot\tfrac13=\tfrac{25}{108},
\]
and the second anchor identity \(\phi_i=\phi_3+(p_i-p_3)J_i\) gives the same values,
\(\tfrac{23}{216}+\tfrac13\cdot\tfrac56=\tfrac{83}{216}\) and \(\tfrac{23}{216}+\tfrac16\cdot\tfrac34=\tfrac{25}{108}\).
For large \(n\) the integrals \(B_i\) are no longer elementary, and the estimators of
Section~\ref{sec:error} replace them, with certified error bounds.

\subsection{Application: risk-weighted patching}\label{SecRiskPatching}

Interpret the three players as a web server, a database server and an IoT sensor, with \(p_i\) the
probability that device \(i\) is compromised, independently. The Shapley values above sum to the
probability \(\tfrac{13}{18}\) that at least one device is compromised, so each is that device's
fair share of the total risk, and they rank the web server first. They can support a triage rule,
although operational decisions should also account for exploitability, exposure, criticality and
dependencies, which this independent-probability model does not represent.

If a new vulnerability raises the database server's probability to \(\tfrac12\), then
\(V=\tfrac{19}{24}\) and \(\phi=(\tfrac{25}{72},\tfrac{25}{72},\tfrac7{72})\). The two servers now have
equal values, and the IoT sensor's share falls although its own probability is unchanged, because
its contribution is more often redundant. Recomputing the vector costs \(O(n^2)\) by exact
quadrature, or \(O(n)\) after one exact value by the anchor scheme.

\subsection{Comparison with permutation sampling and the FFT--multipoint method}
\label{SecOverallComparison}

\subsubsection{Permutation sampling}\label{SecCastroOps}

Castro et al.\ (2009, 2017) estimate \(\phi_i\) by the average
\(\widehat\phi_i=\frac1m\sum_{r=1}^m\Delta_i^{\pi_r}\) of the marginal contributions of player \(i\)
over \(m\) random permutations \(\pi_r\). Their benchmarks (Section~\ref{SecScope}, item~2) are not product games, so our methods do not
apply to them; for those games, as for any game whose marginal contributions can be updated in
\(O(1)\) along a permutation, sampling costs \(O(mn)\). For a product game the marginal
contribution is \(\Delta_i^\pi=p_i\prod_{j\in P_i^\pi}(1-p_j)\), where \(P_i^\pi\) is the set of players
preceding \(i\), and with a running product \(q_k=q_{k-1}(1-p_{\pi(k)})\), \(q_0=1\), each
position costs one multiplication for \(p_{\pi(k)}q_{k-1}\), one subtraction and one multiplication
for the update, and one addition to the accumulator: \(4n\) operations per permutation, \(4mn\) in
total, excluding permutation generation.

The estimator is unbiased, with standard error \(O(m^{-1/2})\). Since \(0\le\Delta_i^\pi\le1\),
Hoeffding's inequality and a union bound over the players give
\(\Pr\bigl(\max_i|\widehat\phi_i-\phi_i|\ge\varepsilon\bigr)\le2n\,e^{-2m\varepsilon^2}\), so
\[
m\ \ge\ \frac{1}{2\varepsilon^2}\log\frac{2n}{\delta}
\]
permutations suffice for accuracy \(\varepsilon\) with probability at least \(1-\delta\). This is a
conservative sufficient condition; variance reduction such as stratification improves practical
performance, but the error remains a Monte Carlo error controlled only in probability.

\subsubsection{The FFT--multipoint method}\label{SecFFTMultipoint}

Write \(P(t)=\prod_{j=1}^n(1-p_jt)=\sum_{r=0}^{n}a_rt^r\). Synthetic division by \(1-p_it\) and
term-by-term integration give
\[
\phi_i=p_i\,R(p_i),\qquad
R(x)=\sum_{s=0}^{n-1}c_sx^s,\qquad
c_s=\sum_{r=0}^{n-1-s}\frac{a_r}{r+s+1}.
\]
The coefficients \(a_r\) follow from a product tree with FFT multiplication in \(O(n\log^2n)\); the
\(c_s\) are a Hankel matrix--vector product, one convolution, in \(O(n\log n)\); and
\(R(p_1),\dots,R(p_n)\) follow by fast multipoint evaluation in \(O(n\log^2n)\). The arithmetic
complexity is therefore \(O(n\log^2n)\), with large implementation-dependent constants. In exact
arithmetic the method is exact; in floating point it is not stable, as the example below shows.

\subsubsection{Gauss--Legendre quadrature}\label{SecGLQuad}\label{SecOtherMethods}

Each integrand \(\prod_{j\ne i}(1-p_jt)\) has degree at most \(n-1\), so \(K\)-point
Gauss--Legendre quadrature is exact in exact arithmetic whenever \(2K-1\ge n-1\), that is,
\(K\ge\lceil n/2\rceil\). Evaluating each integrand separately costs \(O(n^2)\) per player and
\(O(n^3)\) for the vector. Prefix and suffix products remove the extra factor: at node \(t_k\), with
\(b_j=1-p_jt_k\), the products \(L_i=\prod_{j<i}b_j\) and \(R_i=\prod_{j>i}b_j\) follow from one
forward and one backward pass, and \(\prod_{j\ne i}b_j=L_iR_i\). The whole vector then costs
\(3n^2-2n\) arithmetic operations for even \(n\) and \(K=n/2\), with no divisions
(Section~\ref{sec:opcounts}). This standard technique is the exact baseline of this paper; it also
supplies the exact \(\phi_1\) used by the anchor scheme.

\subsubsection{Numerical comparison for \(n=100\)}\label{SecAccuracy100}

We use the following \(100\) distinct probabilities, in increasing order:
\begin{verbatim}
p = [
    0.060831, 0.068504, 0.086031, 0.092507, 0.095824,
    0.098262, 0.099461, 0.108493, 0.112959, 0.121616,
    0.131517, 0.142485, 0.148058, 0.154142, 0.170585,
    0.210370, 0.214280, 0.216343, 0.230656, 0.232472,
    0.239632, 0.256090, 0.265233, 0.272818, 0.276723,
    0.295638, 0.298074, 0.315044, 0.333388, 0.351433,
    0.362569, 0.365265, 0.373013, 0.385402, 0.390274,
    0.398708, 0.400861, 0.401548, 0.404186, 0.405166,
    0.410960, 0.417131, 0.424109, 0.437951, 0.439063,
    0.457540, 0.464830, 0.480770, 0.483201, 0.487399,
    0.511732, 0.523957, 0.528956, 0.529214, 0.540949,
    0.544703, 0.562050, 0.564606, 0.570737, 0.589982,
    0.595965, 0.596613, 0.605063, 0.615379, 0.618061,
    0.644500, 0.656454, 0.661967, 0.664993, 0.665611,
    0.682256, 0.683788, 0.692222, 0.700362, 0.702283,
    0.702730, 0.709271, 0.716438, 0.724336, 0.730995,
    0.738329, 0.746819, 0.751249, 0.782138, 0.814799,
    0.824387, 0.825760, 0.831703, 0.832889, 0.834478,
    0.856658, 0.866439, 0.872774, 0.873917, 0.875774,
    0.894553, 0.901711, 0.917850, 0.946070, 0.949484
]
\end{verbatim}
The reference is the exact Shapley vector, computed in rational arithmetic from the polynomial
expansion of Section~\ref{sec:opcounts}; its values lie between \(0.0012\) and \(0.019\). Errors
are measured by \(E_{\max}=\max_i|\widehat\phi_i-\phi_i|\) and
\(E_{\mathrm{RMSE}}=\bigl(\frac1n\sum_i(\widehat\phi_i-\phi_i)^2\bigr)^{1/2}\).

\paragraph{Permutation sampling.} One run with random seed \(42\) gave the errors in
Table~\ref{tab:product-sampling-errors}, at a cost of \(4mn\) arithmetic operations, \(4\times10^7\)
for \(m=10^5\). These are empirical errors of one run, not guaranteed bounds; they decrease roughly
like \(m^{-1/2}\).

\begin{table}[htbp]
\centering
\caption{Observed errors of permutation sampling in the \(100\)-player product game (one run).}
\label{tab:product-sampling-errors}
\small
\begin{tabular}{rrrr}
\toprule
\(m\) & Operations \(4mn\) & \(E_{\max}\) & \(E_{\mathrm{RMSE}}\) \\
\midrule
\(1{,}000\) & \(4\times10^5\) & \(0.00571231\) & \(0.00183545\) \\
\(10{,}000\) & \(4\times10^6\) & \(0.00209662\) & \(0.000750298\) \\
\(100{,}000\) & \(4\times10^7\) & \(0.000790727\) & \(0.000223456\) \\
\bottomrule
\end{tabular}
\end{table}

\paragraph{Deterministic methods.}
\begin{itemize}
\item \emph{Gauss--Legendre quadrature with prefix and suffix products} (\(K=50\)), in double
precision: \(E_{\max}=1.1\times10^{-15}\), rounding error only, at \(3n^2-2n=29{,}800\) operations.
\item \emph{Anchor scheme}, with the smallest probability \(p_1=0.060831\) as the anchor and \(\phi_1\)
by Gauss--Legendre quadrature. Since \(48.85\le S_i\le49.73\), every player is in Case~2, and the
cost is \(\tfrac32n^2+26n-28=17{,}572\) operations and \(99\) exponentials. The errors are
\(E_{\max}=3.4\times10^{-7}\), \(E_{\mathrm{RMSE}}=1.5\times10^{-7}\) and maximum relative error
\(1.8\times10^{-5}\); for every player the error lies within the certified bound
\(p_i|p_i-p_1|\bigl(16m_i^2/S_i^4+(S_i+1)e^{-S_i/2}/S_i^2\bigr)\) of Section~\ref{sec:err-large},
whose largest value is \(2.1\times10^{-6}\).
\item \emph{FFT--multipoint method}, coefficient form, in double precision: \(E_{\max}=0.23\), larger
than the Shapley values themselves. The coefficients \(a_r\) alternate in sign and reach
\(7.6\times10^{15}\) in absolute value, and evaluating \(p_iR(p_i)\) even from the exact \(a_r\)
rounded to double precision, without any fast algorithm, loses all accuracy; fast multipoint
evaluation adds further instability. Its \(O(n\log^2n)\) bound is thus a statement about exact
arithmetic, and in floating point the method needs a working precision that grows with \(n\).
The value \(n\log_2^2n\approx4{,}414\) at \(n=100\) is not an operation count, and no advantage over
the \(O(n^2)\) methods should be expected at this size.
\end{itemize}

\paragraph{Effectiveness.} Against the sampling run with \(m=10^5\), the anchor scheme is more than
\(2{,}000\) times more accurate at less than \(1/2{,}000\) of the arithmetic cost, and its error is
certified. To \emph{guarantee} its certified accuracy \(\varepsilon=2.1\times10^{-6}\) with probability
\(0.95\), the Hoeffding condition above would require \(m\approx9.4\times10^{11}\) permutations, about
\(4\times10^{14}\) operations; even extrapolating the observed \(m^{-1/2}\) decay, matching the actual
error \(3.4\times10^{-7}\) would take about \(5\times10^{11}\) permutations. Against the exact baseline,
the anchor scheme costs about \(40\%\) less.

\begin{table}[htbp]
\centering
\caption{Comparison of permutation sampling with the deterministic methods of this paper. The last
two columns refer to the product game with \(n=100\) of Section~\ref{SecAccuracy100}.}
\label{tab:overall-comparison}
\footnotesize
\setlength{\tabcolsep}{4pt}
\begin{tabular}{@{}>{\raggedright\arraybackslash}p{0.21\textwidth}
>{\raggedright\arraybackslash}p{0.17\textwidth}
>{\raggedright\arraybackslash}p{0.21\textwidth}
>{\raggedright\arraybackslash}p{0.17\textwidth}
>{\raggedright\arraybackslash}p{0.15\textwidth}@{}}
\toprule
Method & Games & Arithmetic operations & Operations, \(n=100\) & Max.\ abs.\ error, \(n=100\) \\
\midrule
Permutation sampling (Castro et al.)
& any game with a value oracle
& \(4mn\), plus permutation generation
& \(4\times10^7\) (\(m=10^5\))
& \(7.9\times10^{-4}\) (one run; Monte Carlo) \\
\addlinespace
Gauss--Legendre, prefix/suffix (exact baseline)
& product games
& \(3n^2-2n\)
& \(29{,}800\)
& \(1.1\times10^{-15}\) (rounding) \\
\addlinespace
Anchor scheme, Case~2 (this paper)
& product games, large \(S_i\)
& \(\tfrac32n^2+26n-28\), plus \(n-1\) exponentials
& \(17{,}572\), plus \(99\) exponentials
& \(3.4\times10^{-7}\) (certified bound \(2.1\times10^{-6}\)) \\
\addlinespace
FFT--multipoint
& product games
& \(O(n\log^2n)\)
& no exact count
& \(0.23\) in double precision \\
\bottomrule
\end{tabular}
\end{table}

All the deterministic methods are free of Monte Carlo error once the probabilities are given. If
the \(p_i\) are themselves estimated from data, their statistical uncertainty is a separate error,
not to be confused with the arithmetic or approximation error of the Shapley computation.

\newpage

\section{Appendix I. An integral and combinatorial approach to derive the Shapley value formula}  \label{Sec1}
Let $E$ be a finite set of $n$ players, and $v: \mathcal{F} \rightarrow \mathbb{R}$ be a real-valued function defined on
the power set $\mathcal{F}$ of $E$, satisfying the boundary condition $v(\emptyset)=0$. (Super-additivity,
$v(S \cup T) \geq v(S) + v(T)$ whenever $S \cap T = \emptyset$, is often assumed but is not needed below.)

\smallskip
\smallskip

\smallskip
\smallskip

For any $R \subseteq E, R \neq \emptyset$, Shapley defined the unanimity game
$$
v_R(S)=\left\{
\begin{array}{lc}
	1  &  {\rm if\;} S \supseteq R \\
	0  & {\rm \;\: otherwise}.
\end{array}
\right.
$$

The Shapley value of player $i$ with respect to such a unanimity game $v_R$ is given by
(Shapley 1953b, Lemma 2)

$$
\phi_i(v_R)=\left\{
\begin{array}{lc}
	\frac{1}{r}  & \hspace*{-0.2in}  {\rm if\;} i \in R \\
	0  & {\rm  \;\; if\;} i \in E \setminus R.
\end{array}
\right.
$$

In the above formula, lower case letter $r$ stands for the number of players in $R$. This notation will be also used with other letters, as needed.

Analogous to the usual expansion of periodic functions into Fourier series, i.e., linear combinations of trigonometric functions, Shapley showed that real-valued set functions, i.e., characteristic functions of games, can be expanded into linear combinations of unanimity games (Shapley 1953b, Lemma 3), i.e.,

$$v(S) = \sum_{\emptyset \neq R \subseteq E} c_R(v) v_R(S), S \subseteq E,$$ where
$$\hspace{2in} c_R(v) = \sum_{T \subseteq R} (-1)^{r-t}v(T)    \hspace{2in} (8) $$ is the M\"{o}bius coefficient with respect to the unanimity game $R$, a name that
Shapley did not explicitly call in his 1953 paper. M\"{o}bius transform defined on locally finite posets was extensively and intensively studied by Rota (Rota 1964). In particular, M\"{o}bius transform is widely applied to  Boolean lattices such as the power set of a given set with inclusion relation between subsets (Grabisch and Labreuche 2005).

It follows from Lemmas 2 and 3 that Shapley obtained his initial representation of Shapley value for player $i$:

$$\hspace{1.9in} \phi_i(v) = \sum_{i \in R \subseteq E} \frac{c_R(v)}{r}, \;\;i \in E. \hspace{1.75in} (10)$$

Then he stated that ``Inserting (8) and simplifying the result gives us''

$$\hspace{0.6in} \phi_i(v) = \sum_{i \in S \subseteq E} \frac{(s-1)! (n-s)!}{n!} v(S)
-\sum_{i \notin S \subseteq E} \frac{s! (n-s-1)!}{n!} v(S), \; i \in E. \hspace{0.6in} (11)$$
The two sums are often put together into one as follows:

$$\hspace{1.35in} \sum_{i \notin S \subseteq E} \frac{s! (n-s-1)!}{n!} \Big {[} v(S \cup \{i\}) - v(S) \Big {]}, \; i \in E. \hspace{1.38in} (12)$$

However, the omitted simplification can be of great interest for reasons like mathematical rigor, alternative insights, or just curiosity.
In the current literature, the derivation of Formula (12) is  often through permutations or  combinations, $s!$ permutations before the $(s+1)$-th position and $(n-s-1)!$ permutations after this position.
We have not found a direct counting from Formula (10) in the literature. As such, one purpose of this article is to provide a rigorous integral and combinatorial approach to establish Formula (12).

\subsection {An integral and combinatorial approach to Shapley value calculation}
In this section, for Shapley's claim, ``Inserting (8) and simplifying the result gives us Formula (11)'',
we will perform the calculation and simplification through a rigorous definite integral and the Binomial Theorem.

First by plugging  Formula (8) into Formula (10), we get

$$\phi_i(v) = \sum_{i \in R \subseteq E} \; \frac{1}{r} \; \sum_{T \subseteq R} (-1)^{r-t} v(T)$$
$$ = \sum_{i \in R \subseteq E} \; \frac{1}{r} \;
\Big {[}\sum_{i \in T \subseteq R} (-1)^{r-t} v(T) + \sum_{i \notin T \subseteq R} (-1)^{r-t} v(T) \Big {]}
$$

Let $R^\prime = R \setminus \{i\}$ and $T^\prime = T \setminus \{i\}$. Accordingly, it holds that
$r^\prime = r-1$ and  $t^\prime = t -1$. Then we have

$$\phi_i(v) = \sum_{R^\prime \subseteq E \setminus \{i\}} \frac{1}{r^\prime +1}
\sum_{T^\prime \subseteq R^\prime} (-1)^{(r^\prime +1) -(t^\prime+1)}
v(T^\prime \cup \{i\}) + \sum_{R^\prime \subseteq E \setminus \{i\}} \frac{1}{r^\prime +1}
\sum_{T \subseteq R^\prime}(-1)^{(r^\prime+1)-t} v(T)
$$

By changing $T^\prime$ to $S$ in the first sum and $T$ to $S$ in the second sum,
we get

$$\phi_i(v) = \sum_{R^\prime \subseteq E \setminus \{i\}} \frac{1}{r^\prime +1 }
\sum_{S \subseteq R^\prime} (-1)^{(r^\prime +1) -(s+1)} v(S \cup \{i\})
+ \sum_{R^\prime \subseteq E \setminus \{i\}} \frac{1}{r^\prime +1 }
\sum_{S \subseteq R^\prime} (-1)^{(r^\prime +1) -s} v(S)$$
$$=\sum_{R^\prime \subseteq E \setminus \{i\}} \frac{1}{r^\prime +1 }
\sum_{S \subseteq R^\prime} (-1)^{r^\prime - s} \Big {[} v(S \cup \{i\}) - v(S) \Big {]}$$

Next, interchanging the order of the two sums, in the new inner sum first adding the terms that
correspond to the coalitions of a same size $k$ to obtain partial sums,  and then putting all these partial
sums together, we have

$$\phi_i(v) = \sum_{S \subseteq E \setminus \{i\}} \sum_{S \subseteq R^\prime \subseteq E \setminus \{i\}}
\frac{1}{r^\prime +1} (-1)^{r^\prime -s} \Big {[} v(S \cup \{i\}) - v(S) \Big {]} $$

$$= \sum_{S \subseteq E \setminus \{i\}} \sum_{k=s}^{n-1}
\sum_{S \subseteq R^\prime \subseteq E \setminus \{i\}, r^\prime = k} \frac{1}{r^\prime+1}(-1)^{r^\prime -s} \Big {[} v(S \cup \{i\}) - v(S) \Big {]}$$

\begin {equation} \label{1}
= \sum_{S \subseteq E \setminus \{i\}} \sum_{k=s}^{n-1} \binom{n-1-s}{k-s}
\frac{1}{k+1}(-1)^{k-s} \Big {[} v(S \cup \{i\}) - v(S) \Big {]}.
\end{equation}

Let $l = k-s$, then
\begin {equation} \label{2}
\sum_{k=s}^{n-1} \binom{n-1-s}{k-s} \frac{1}{k+1}(-1)^{k-s}
= \sum_{l=0}^{n-1-s} \binom{n-1-s}{l} \frac{1} {l+s+1} (-1)^l.
\end{equation}
Denote this combinatorial expression by $S_{n,s}$.

\bigskip

We next evaluate $S_{n,s}$
using a definite integral representation and the Binomial Theorem.
\bigskip

Substituting $\frac{1}{l+s+1} = \int_0^1 x^{l+s} \, dx$ in $S_{n,s}$ gives

$$
S_{n,s} = \sum_{l=0}^{n-1-s} \binom{n-1-s}{l} (-1)^l \int_0^1 x^{l+s} \, dx.
$$

Interchanging the sum and the integral gives
$$
S_{n,s} = \int_0^1 x^s \sum_{l=0}^{n-1-s} \binom{n-1-s}{l} (-1)^l x^l \, dx.
$$

The inner sum is the binomial expansion of $(1-x)^{n-1-s}. $
Thus,
\begin {equation} \label{S-n-s}
S_{n,s} = \int_0^1 x^s (1-x)^{n-1-s} \, dx,
\end{equation}
which is the value of the standard Beta function, i.e., $S_{n,s} = Beta(s+1, n-s).$

Now by using the Beta-Gamma relation
$Beta(a, b) = \frac{\Gamma(a) \Gamma(b)}{\Gamma(a+b)}$,
we obtain
$$
Beta(s+1, n-s) = \frac{\Gamma(s+1) \Gamma(n-s)}{\Gamma(n+1)}.
$$
Since $\Gamma(s+1) = s!$ and $\Gamma(n+1) = n!$, we get
$ Beta(s+1, n-s) = \frac{s! (n-1-s)!}{n!}.$

\bigskip

Therefore, the following is established:
$$
\sum_{l=0}^{n-1-s} \binom{n-1-s}{l} \frac{(-1)^l}{l+s+1} = \frac{s! (n-1-s)!}{n!}.
$$
This derives Shapley's Formula (12) using the definite integral and combinatorial approach, and hence provides a supplementary calculation to the omitted part in Shapley's original paper.

\bigskip

\noindent
{\bf Acknowledgement:}
The authors thank two anonymous reviewers for indicating the connection of the inhomogeneous Bernoulli trials game to the broader class of cooperative product games, and for the suggestion to compare the present results with existing studies, in particular those of Castro et al.

\end{document}